\documentclass[12pt]{article}
 
\usepackage{amssymb}
\usepackage{amsmath}
\usepackage{amsfonts}
\usepackage{amscd}
\usepackage{graphicx}
\usepackage{amsthm}
\usepackage{setspace}
\usepackage{epsfig} 
\usepackage{bm}
\usepackage{amssymb}
\usepackage{hyperref}
\usepackage[authoryear,round]{natbib}
\usepackage{booktabs}
\usepackage{threeparttable}
\usepackage{float}
\usepackage{subfig}

\theoremstyle{definition} \newtheorem{theorem}{{\bf \sc Theorem}}
\theoremstyle{definition} \newtheorem{proposition}{{\bf \sc Proposition}}
\newtheorem{lemma}{{\bf \sc Lemma}}

\theoremstyle{definition}

\newtheorem{definition}{{\bf \sc Definition}}

\newtheorem*{conjecture*}{Conjecture}

\newcommand{\be}{\begin{equation}}
\newcommand{\ee}{\end{equation}}
\newcommand{\bes}{\begin{equation*}}
\newcommand{\ees}{\end{equation*}}

\newcommand{\R}{\mathbb R}

\renewcommand{\Pr}{\mathbb P}
\newcommand{\ind}{\mathbf 1}
\newcommand{\calI}{\mathcal I}

\DeclareMathOperator{\U}{U}

\allowdisplaybreaks 

\begin{document}

\title{Chaos and Unraveling in Matching Markets\thanks{For helpful comments and suggestions we thank Yossi Feinberg, Alex Frankel, Ben Golub, John Hatfield, Fuhito Kojima, Eli Livne, Muriel Niederle, Michael Ostrovsky, Andrzej Skrzypacz and Bob Wilson.}} 
\author{Songzi Du\thanks{\mbox{Graduate School of Business, Stanford University.
Email: songzidu@stanford.edu,}} \and %
 \and Yair Livne\thanks{\mbox{Graduate School of Business, Stanford University.
Email: ylivne@stanford.edu}}}
\date{\today \\
First draft February 18, 2010}

\maketitle

\markboth{Songzi Du and Yair Livne}{Chaos and Unraveling in Matching Markets}

\begin{abstract}
We study how information perturbations can destabilize two-sided matching markets. In our model, agents arrive on the market over two periods, while agents in the first period do not know the types of those arriving later. Agents already present in the market may match early or wait for the small group of new entrants. Despite the lack of discounting or risk aversion, this perturbation creates incentives to match early and leave the market before the new agents arrive. These incentives do not disappear as the market gets large. Moreover, we identify a new adverse phenomenon in this setting: as markets get large the probability of \emph{chaos} -- where no early matching scheme for existing agents is robust to pairwise deviations -- approaches 1. These results are independent of the distribution of agents' types and robust to asymmetries between the two sides of the market. Our findings thus suggest that matching markets are extremely sensitive to institutional details and uncertainty.
\end{abstract}

\newpage 

\section{Introduction}

Many two sided labor matching markets suffer from unraveling, with matching decisions taken earlier than efficient. Centralized and decentralized matching mechanisms, designed to coordinate matching times among different agents and prevent unraveling, often collapse.\footnote{See for example \cite{RothXing} for many examples of markets suffering from unraveling and collapse of centralized matching mechanisms.} The economic literature suggests several particular forces that can drive unraveling: risk aversion on the side of agents, similar preferences, exploding offers or unbalanced supply and demand. But, does unraveling arise in simpler markets perturbed only by a small uncertainty facing agents? Can markets fail even when they grow large and this uncertainty becomes trivial?

To explore these issues, we study in this paper an information perturbation of a simple matching market. In this market, $n$ men and women arrive on the market in the first period, having types drawn independently according to some distribution. The types of men reflect the fully-aligned preferences of women over possible matches, and vice versa. In this environment, a fully assortative match is the only possible stable outcome. We perturb this simple setup by introducing a second period in which $k$ new agents will arrive on the market on each side, with types drawn from the same distribution. We assume that in the second period an assortative match will take place, perhaps as the result of a centralized matching mechanism. Anticipating the arrival of the new agents, men and women in the first period may match and leave the market permanently. However, agents do not have an obvious incentive to match early -- we assume no discounting or risk aversion by agents.

We first show that in small markets two adverse phenomena arise. The first is unraveling -- given that all other agents wait for the second period, there may be a positive probability that some pairs of agents will want to match early and leave the market. In our setting, this happens when agents on both sides of the market are positioned high relative to their surroundings, and thus for these agents waiting for new entrants has a larger downside than upside.

The second phenomenon is \emph{chaos}. We show that in some realizations there exists no possible pure strategy early matching scheme, where some couples leave the market early while some stay for the second period, which is robust to pairwise deviations. That is, for every possible partial pairing of couples in the first period, given that all paired couples leave the market in the first period while all other agents remain for the second period, there exists either: (1) an agent who is part of a pairing but would rather break of his engagement and stay for the second period or (2) a pair of agents who are unpaired but would rather leave the market together. Chaos stems from the externalities which different couples impose on one another in this perturbed market -- in the second period lower ranked agents on one side of the market provide insurance for higher ranked agents on the other side, while higher ranked agents on one side of the market cap the upside for lower agents on the other side from staying on the market. Since this definition also covers the empty partial pairing as a special case, chaotic realizations are a subset of the ones displaying unraveling.

Our main result is that as initial markets grow thick, holding the size of the number of entering agents fixed, chaotic realizations occur with a probability approaching one. This result is robust to the specification of the distributions from which agents' types are drawn, asymmetries between the two sides of the market and different specification of the exact game that takes place in the first period. Thus, this result implies that small uncertainties and seemingly minor institutional details in matching markets should be of concern to the market designer. In practice, phenomena like chaos are likely to destabilize matching markets over time, especially centralized matching mechanisms. Moreover, thick markets do not mitigate, and perhaps exacerbate these instability problems.

A second result answers the question of how strong unraveling forces are in this perturbed market -- that is, assuming that all other agents remain for the second period, how likely are pairs of agents to prefer matching early? We prove that as markets grow thick, while the number of arriving agents is fixed, each couple in the interior of the distribution has a probability approaching 25\% to prefer to match early, and this is regardless of the distributions of types on the two sides of the market. Moreover, we show that with a probability approaching 1, a positive share of agents \emph{simultaneously} prefers to unravel. This suggests that in markets perturbed by uncertainties, having all agents wait for a central match maybe ``very far" from an equilibrium, especially if the market is large.

Unraveling in matching markets has previously been linked in the literature to several different causes. \cite{LiRosen}, \cite{Suen2000} and \cite{LiSuen} study unraveling in environments where matching early acts as insurance for risk averse workers. \cite{Halaburda} shows that similar preferences by firms can increase the probability of unraveling in labor markets. Although our setting features identical preferences across agents, which in Halaburda's paper are the most conductive for unraveling, our setting is different as in it all agents who are on the market have full information regarding their type and the preferences of all other agents. Thus, each agent can rather precisely assess his own chances on the market in future periods.

\cite{RothXing} discuss instability in centralized matching mechanisms and suggest the presence of market power as a possible reason for unraveling. \cite{DamianoLiSuen} explore search costs as a reason for unraveling. \cite{NiederleRoth} show in a theoretical and an experimental setting that exploding offers and binding contracts promote early contracting in markets. In a recent paper, \cite{Fainmesser} shows that the local micro-structure of social networks through which information about agents' type flows, can influence whether a matching market will unravel or not. \cite{NiederleRothUnver} discuss the role, theoretically and experimentally, of the balance between demand and supply in unraveling.

The remainder of the paper is organized as follows: in Section \ref{sec:model} we present the basic model and discuss its generality; Section \ref{sec:smallmarkets} shows how chaos and unraveling arise in small markets and discusses the basic forces driving these phenomena; Section \ref{sec:chaos} contains the main theorem proving the prevalence of chaos in large, perturbed markets and outlines a roadmap for the proof; Section \ref{sec:unraveling} presents the main theorem on unraveling incentives for agents in perturbed markets and Section \ref{sec:conclusion} concludes.

\section{The Basic Model} \label{sec:model}

We study a simple two-sided matching market in which men and women arrive on the market and are matched across two periods.

\paragraph{Players} In the first period $n$ men and $n$ women are present on each side of the market. After the first period is over, $k$ new men and $k$ new women arrive on the market.  The new arrivals can be interpreted as agents whose types are unknown by others in the first period.  In some applications it can also be interpreted literally.

Both men and women have types. Men have types independently and identically drawn (i.i.d.) from some distribution $F$ on the interval $[\underline{m}, \overline{m}]$ which has density $f$; and the women have types drawn i.i.d.\ from some distribution $G$ on the interval $[\underline{w}, \overline{w}]$ which has density $g$. In the first period, the types of all men and women present in the market are known to all agents. 

For each $1\leq i \leq n$ we denote by $m_i$ the $i$-th lowest type among the $n$ men present in the first period.  Likewise, women's types in the first period are denoted by $w_n \geq \ldots \geq w_1$.

For most of this paper, we focus on the case when $n$ is large in comparison to $k$, i.e., there is only a small amount of uncertainty in period 1.

\paragraph{Preferences} Both men and women have utility functions which are a linear function of the type of the agent they are matched with, and some function of their own type. Thus, men and women have fully aligned preferences: men's types fully reflect every women's preferences over men, and vice versa. 

We assume that agents are risk neutral and have no discounting between the two periods.

\paragraph{Game Order} In the first period men and women may form pairs and permanently leave the market, gaining utility from their match partner. We intentionally do not explicitly model the exact game played in this period but only require, in the spirit of the stability literature in matching, that pairs of agents can make jointly deviate. The purpose is to have the result be as general as possible, as we discuss in the next item.

In the second period all players who are present --- including players who stayed on the market in the first period and the new entrants --- are assortatively matched, i.e., the highest type man is matched to the highest type woman, the second highest to the second highest, etc.  Given the aligned preferences structure, assortative matching is the unique stable outcome in the second period. Furthermore, assortative matching is the unique efficient outcome if utility functions admit complementarities in the types of both agents involved in a match.

\paragraph{Comments on Generality} The setting we present is general across several dimensions:

\begin{itemize}
	\item 	{\bf Game-play}.  We intentionally do no model the game play in the first period. All we require is that pairs of agents, of the same ranking or otherwise, have the ability to match and leave the market in the first period. This can be manifested through different offer-response mechanisms, which can be sequential, one-shot etc.
	\item 	{\bf Information}. The assumption on common-knowledge of all types present in the first-period market is stronger than we need.  For our unraveling results, it is enough that agents know their own type, and the types of agents on the other side of the market within a distance of $k$ of their own rank.  In a market with $n=20$ and $k=2$ for example, the man ranked $13$ needs only to know the types of women of ranks $11$ through $15$. If $n$ is large relative to $k$, this implies that agents need only ``local'' knowledge of the market.

For our chaos result (which we require that $k=1$), what we actually need is that agents know their own type, and know, among those \emph{who stay to the second period}, the types of the two agents on the other side of the market closest to their own rank.  Since the decision of whether to match early or to stay for the second period is endogenous, it's convenient for us to assume that every agent knows all other agents' type.  

However, as we will show in our result, it is highly unlikely that a large number of consecutively ranked agents would simultaneously decide to match early.  Suppose that a man $m$ knows the types of all women within a distance of $C$ of his own rank, where $C$ is a large constant.  Then with very high probability that there will be two women, one above and one below $m$'s rank, who are among the $2C+1$ women that $m$ know and who decide to stay to the second period.  Therefore, the information assumption for our chaos result will hold with very high probability if agents know the types of agents on the other side of the market within a distance of $C$ of their own rank, where $C$ is a large constant that does not grow with $n$.

	\item 	{\bf Preferences}. Utility functions are unrestricted as long as they are a linear function of the type with which agents are matched. This can captures markets with complementarities in types and ones with substitutes.
	
	\item 	{\bf Distributions} Our results depend only on regularity conditions on $F$ and $G$, but are otherwise distribution-free. This generality allows the model do describe, for example, asymmetric rates of entry into the market from the two sides. This could be modeled by having one of the distributions have an very high density --- simulating an atom --- near the lowest type of agent.
\end{itemize}

This generality serves to show that the adverse effects we identify due to informational perturbations are relevant to many settings, and are fundamental to the structure of matching markets.

\section{Chaos and Unraveling in Small Markets} \label{sec:smallmarkets}

We start by examining agents' behavior in small markets. We are particularly interested in cases where having all agents wait to be matched in the second period is not robust to pairwise deviations agents may make in the first period. Any deviation from the full assortative match may lead to inefficient outcomes, and moreover may cause any institutional matching mechanism which exist in the market to collapse (cite someone).

Agents may prefer to leave the market early in the first period, but initially, it is not entirely clear why unraveling should arise in our setting --- agents have no inherent preferences for finding and early match and are not risk averse so do not mind the uncertainty associated with waiting. 

Consider the simplest possible environment in our setting, where $n=k=1$, so in the first period a single man and a single woman arrive on the market, and in the second period another pair will arrive. For simplicity, assume that agents' types (for both the men and the women) are distributed according to the uniform distribution on the unit interval $[0,1]$.

Regardless of the exact game form played in the first period, we can ask which pair-types would prefer to leave the market early and match with each other rather than to wait for the second pair to arrive and match assortatively.

A man of type $m$ would prefer to leave the market in the first period with a woman of type $w$ if his expected match in the second period was lower than $w$. This future match can fall into three cases:
\begin{enumerate}
	\item 	If the both entrants are of higher type than the incumbents, or if both entrants are of lower type than the incumbents, then the incumbents are matched in the second period;
	\item 	If the entrant man is of a higher type than $m$ while the entrant woman is of lower type than $w$, then the incumbent man is matched to the entrant woman, and gets a match of quality lower than $w$. These are the realizations in which the man loses from waiting.
	\item 	If the entrant man if of a lower type than $m$ while the entrant woman is of higher type than $w$, then the incumbent man is matched to the entrant woman, and gets a match of quality higher than $w$. These are the realizations in which the man gains by waiting. 
\end{enumerate}
 
Formally, the expected value of this match is given by:
$$\big(mw+(1-m)(1-w)\big)w+(1-m)\int_0^wtdt+m\int_w^1tdt$$
where the terms correspond to the cases above. Comparing this to $w$ and simplifying, we have that the man in the first period would strictly prefer to match early if and only if:
$$(1-m)w^2>m(1-w)^2$$
Symmetrically, the condition for the first period woman to prefer an early match is given by:
$$(1-w)m^2>w(1-m)^2$$

Therefore, pairs $(m,w)$ which satisfy both of the above conditions would rather match early and leave the market over staying for the bigger match in the last period. The set of such pairs is non-empty, and is plotted in Figure \ref{fig:example1}.

	\begin{figure}[ht]
	\begin{center}
		\includegraphics{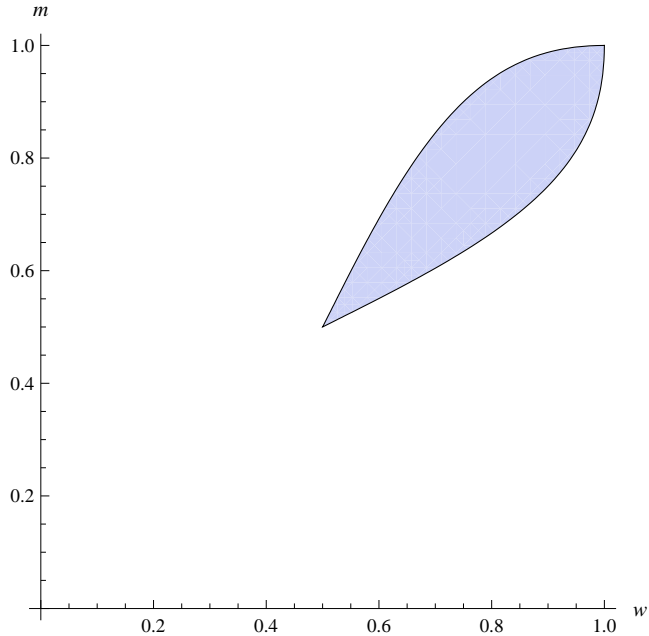}
		\caption{Unraveling in a small market. Types in the shaded area prefer to match early.} \label{fig:example1}
	\end{center}
	\end{figure}

The figure reveals that a non negligible share of pairs --- approximately 9.7\% of pairs in this example --- prefer to match early, suggesting that unraveling is present as a phenomena even in this simple setting.

Pairs that prefer an early match are of two notable characteristics: they are of \emph{similar type} and have relatively \emph{high types}. Intuitively, pairs with similar types have a similar probability of gaining or losing from waiting to the full match. However, since they are of high type and hence the upside from waiting is smaller than the possible downside. Thus, these agents experience an \emph{endogenous risk aversion}, preferring to leave the market. Other pairs prefer to stay --- if the two types are not similar, the higher of the two knows that she/he can expect with high probability a better draw in the second period, and thus prefers to stay. Alternatively, if the two types are similar but both low, then both agents have a higher upside than downside from waiting, since a good draw is expected to improve their match by more than a bad one would.

Following this result, a natural question to ask is whether these forces are driven by the small size of the market and disappear in a thick market, or whether they persist even as the market becomes large. We rigorously explore this question in Section \ref{sec:unraveling}.

Unraveling however, is not the only adverse outcome that can emerge in small markets. Consider now a second example: take $n=2$, $k=1$ and have agents' types again be uniformly distributed on the unit interval. Specifically, consider the realization of types where:

$$m_1=w_1=\frac{2}{5} \qquad \qquad m_2=w_2=\frac{3}{5}$$

What will agents do in equilibrium? First, consider the option of the top-ranked agents staying in the market and waiting for the second period. If this is the case, let us look at the incentives of the lower-ranked pair. Leaving the market together gives each of the two a match of type $2/5$. However, if one of them would choose to stay on the market, their payoff given that the top-ranked pair indeed stays for the second period is $49/125$,\footnote{This payoff comes from the following calculation: the agent gets $2/5$ with probability $(2/5)^2 + (3/5)^2$ (no change in his/her relative ranking in the second period); gets $1/5$ with probability $(3/5)(2/5)$ (matched with the new agent whose type is between $0$ and $2/5$); gets $1/2$ with probability $(2/5)(1/5)$ (matched with the new agent whose type is between $2/5$ and $3/5$); gets $3/5$ with probability $(2/5)^2$ (matched with the top agent of the first period).  Taking expectation over all possible payoffs gives $49/125$.} which is smaller than $2/5$. Thus, the bottom-ranked pair would rather leave the market together. The intuition here is along the same lines as before --- the presence of the top-ranked pair in the second period caps the upside from staying for the bottom-ranked pair in this realization, and thus they rather leave the market early.

If the bottom-ranked pair indeed matches early and leaves the market, then when considering the incentives of the top-ranked agents, we return to the first example discussed above, with $n=k=1$. As displayed in Figure \ref{fig:example1}, a pair with both agents of type $3/5$ prefers to match early and leave the marked. Thus, the top-ranked pair would prefer to match early.

If indeed the top-ranked pair contracts early and leaves the market, when considering the incentives of the lower-ranked pair we are again back to the first example. Returning to Figure~\ref{fig:example1}, we note that a pair of agents both with type $2/5$ would not want to match early.

Finally, we consider the bottom-ranked pair staying on the market for the second period. If this is the case, the top ranked agents may get $3/5$ by contracting early, but have an expected match of $76/125$,\footnote{This is calculated in the same fashion as in footnote 2.} a higher payoff, in the second period. Intuitively, the bottom-ranked pair's presence in the second period market bounds the downside from staying on the market for the top-ranked agents. Thus, if the bottom-ranked pair stays on the market, the top-ranked pair would tend to stay as well.

But now we have come full circle --- as displayed in Figure \ref{fig:example2}, each of choices described by the two pairs implies another by the other pair, but none of these choices is consistent with a corresponding choice by the other pair.\footnote{To complete the argument, one also has to consider possible early matching by cross matches of top-ranked and lower-ranked agents. It is easy to check that in this example these are never profitable to the top ranked agent.}

\begin{figure}[ht]
\begin{center}
	\includegraphics[scale=0.65]{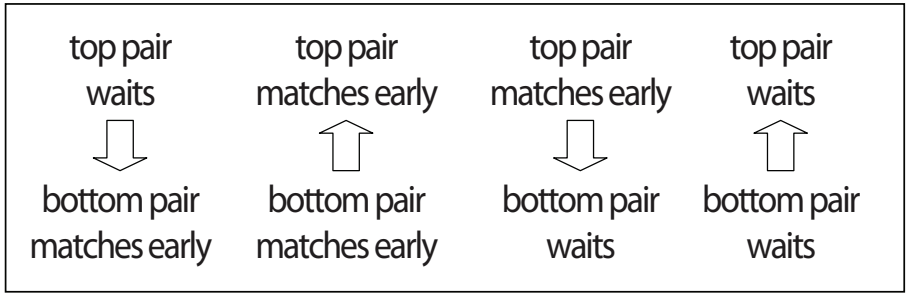}
	\caption{Chaos in a Small Market} \label{fig:example2}
\end{center}
\end{figure}

We will call such a realization \emph{chaotic}, as it admits no preordained, partial early matching which is immune to pairwise deviations. In such a realization, not only will some agents not wait for the (perhaps efficient) second period match, but moreover, in any equilibrium agents will be using mixed strategies when deciding on matchings. This is clearly an undesirable outcome for a market designer in charge of a matching mechanism. 

Chaos originates in this example directly from the informational perturbation represented by the incoming agents. Because of the uncertainty agents are faced with in the second stage, ranked couples impose externalities on each other. In this example, the top-ranked couple imposes a negative externality on the bottom-ranked couple if it decides to stay for the second period, and does so by capping the upside from staying for the lower-ranked couple. On the other hand, the lower-ranked couple imposes a positive externality on the top-ranked couple by staying in the market --- they provide insurance against very negative outcome for the top-ranked types.

The existence of chaotic realizations is perhaps surprising, but in this example, they are rare --- only about \%1 of realizations are chaotic. We again ask what happens when markets grow large --- does thickness smooth out markets to prevent chaotic outcomes, or do they become common? This question is particularly interesting if the informational perturbation becomes small relative to the size of the market. One might may expect that when uncertainty becomes minor, these perverse cycles of externalities may disappear. We turn to answer this question directly in the next section.

\section{Chaos}\label{sec:chaos}

In this section we answer the question posed in the previous section and show that a seemingly minor informational perturbation in the form of one pair of arriving players\footnote{We maintain the assumption of $k=1$ throughout this section.} in the second period has a large rippling effect on the stability of the matching market, \emph{especially when the market is thick}. We do this explicitly by studying the prevalence of chaotic realizations. To define these formally, we first define the notion of an \emph{early matching}.

\begin{definition} \label{def:earlymatching}
An early matching is a function $\mu : \{1, \ldots, n\} \rightarrow \{1, \ldots, n, \emptyset \}$, satisfying $\mu(i) = \mu(j) \neq \emptyset \Rightarrow i = j$, which is a proposed unraveling scheme for the first period. According to $\mu$:
\begin{enumerate}
	\item 	If $\mu(i)\neq \emptyset$, then the man of rank $i$ is designated to leave the market early with the woman of rank $\mu(i)$;
	\item 	If $\mu(i)= \emptyset$, then the man of rank $i$ is designated to stay in the market to be matched in the second period.
	\item If $i \not \in \{\mu(1), \ldots, \mu(n)\}$, then woman of rank $i$ is also designated to stay in the market to be matched in the second period.
\end{enumerate}
\end{definition}

For every particular realization, we are interested in early matchings which players will not want to deviate from. Conditional on the realization, such early matchings may serve as schemes for possible equilibria of the specific game played in the first period.

Formally, let $C = \{ (m_i, w_i) \}_{1 \leq i \leq n}$ be an rank-ordered list of types of men and women present in the first period: $m_n \geq m_{n-1} \geq \ldots \geq m_1$ and $w_n \geq w_{n-1} \geq \ldots \geq w_1$. $C$ fully expresses, up to a permutation, the realization of types in the population. Now, fix an early matching $\mu$.  Let $C(\mu)$ be an ordered list of men and women in $C$ who wait for the second period, according to $\mu$.  For a man of rank $i$ such that $\mu(i) \neq \emptyset$, let $C(\mu, i)$ be the ordered list of men and women in $C$ who wait for the second period \emph{and} the pair $(m_i, w_{\mu(i)})$. The set $C(\mu, i)$ would exactly represent the types in the second period if either man $i$ or woman $\mu(i)$ were to deviate from the unraveling pattern described by the early matching $\mu$. 

For a rank-ordered list $L$ of pairs of types, and a type of man $m$ included in $L$, let $U(m, L)$ be the expected match of a man of type $m$ staying in the market for the second period, when the population of first-period types also staying is expressed by the list $L$. Similarly, define $V(w, L)$ to be the expected match of a woman of type $w$ when she is part of the list $L$ which stays on the market for the second period.

Formally, suppose that $m$ and $w$ are of the same rank in $L$, and that $w_+$ and $w_-$ are the women ranked just above $w$ and just below $w$ in $L$, respectively. Since only one new pair of man and woman arrives in the second period ($k=1$), we then have:
\begin{align}
\label{eq:payoff_U}
U(m, L) \; = \; & \bigg[ F(m) G(w) + (1-F(m)) (1-G(w)) \bigg] w \\
 + \;& (1-F(m)) \left( w_- G(w_-) + \int_{w_-}^w x g(x) \, dx \right) \notag \\
 + \;& F(m) \left( w_+ (1-G(w_+)) + \int_{w}^{w_+} x g(x) \, dx \right), \notag
\end{align}
where the first term represents the events where $m$ and $w$ are matched in the second period, the second term represents events where $m$ is match with a worse type than $w$, and the last term represents event where $m$ is match with a higher type than $w$. Likewise, suppose that $m_+$ and $m_-$ are the men ranked just above and below $m$ in $L$, we have:

\begin{align}
\label{eq:payoff_V}
V(w, L) \; = \; & \bigg[ G(w) F(m)  + (1-G(w))(1-F(m)) \bigg] m \\
 +\; & (1-G(w)) \left( m_- F(m_-) + \int_{m_-}^m x f(x) \, dx \right) \notag \\
 +\; & G(w) \left( m_+ (1-F(m_+)) + \int_{m}^{m_+} x f(x) \, dx \right), \notag
\end{align}

We can now define the notion of a pairwise-stable early matching.

\begin{definition}
\label{def:pairwisestable}
Given an ordered list of first-period types $C = \{ (m_i, w_i) \}_{1 \leq i \leq n}$, an early matching $\mu$ is \emph{pairwise stable} if:
\begin{enumerate}
\item For any couple $(m_i, w_{\mu(i)})$ who matches early ($\mu(i) \neq \emptyset$), we have $U(m_i, C(\mu,i)) < w_{\mu(i)}$ and $V(w_{\mu(i)}, C(\mu,i)) < m_i$.
\item For any man and woman $(m_i, w_j)$ who both stay for the second period ($\mu(i) = \emptyset$ and $j \not \in \{\mu(1), \ldots, \mu(n) \}$), we have either $U(m_i, C(\mu)) \geq w_j$ or $V(w_j, C(\mu)) \geq m_i$.
\end{enumerate}
\end{definition}

The definition requires that a pairwise stable early matching $\mu$ for a realization $C$ would specify a ``stable" unraveling --- (1) every member of a pair which is designated by $\mu$ to match early would prefer the designated match over staying for the second period, and (2) every pair of agents designated to stay on the market by $\mu$ does not want to deviate and match early. These preferences are formed assuming that all other players act according to $\mu$.

As displayed in Section \ref{sec:smallmarkets}, stable matching functions may not exist for some realizations due to the externalities pairs' decisions impose on other pairs. We now formally define a \emph{chaotic} realization to be one which admits no pairwise-stable early matching. 

Surprisingly, as markets grow thick and the impact of the informational perturbation in terms of utility becomes small, we find that chaotic realizations happen with a probability approaching one.

\begin{theorem} \label{thm:chaos}
Suppose that (i) there exists $p \in (0, 1)$ such that the densities $f$ and $g$ are positive and continuous at $x$ and $y$, respectively, where $F(x) = G(y) = p$ and (ii) there exist positive and finite numbers $\underline{a}, \overline{a}$ such that $\underline{a} \leq f(m), g(w) \leq \overline{a}$ for all $m \in [\underline{m}, x]$ and $w \in [\underline{w}, y]$.  Then, the probability of a chaotic realization in the first period tends to 1 as $n$ tends to infinity.
\end{theorem}


We remark that only assumption (i) is crucial: (ii) is made only for convenience, and the theorem should continue to hold under a much weaker version of (ii).

The theorem says that when the market in the first period is sufficiently thick, then with probability close to 1 \emph{any} proposed unraveling arrangement (including no unraveling) will be unstable: either a couple would have incentives to jointly deviate from waiting for the second period by matching early, or an individual man or woman would have an incentive to deviate from his/her early matching by staying for the second period.

We emphasize that it is the size of the market that is responsible for the magnitude of the instability: for $F=G=U[0, 1]$ and when $n = 2$, numerical integrations reveal that the probability of realized types having a stable early matching is $\approx 0.99$; when $n = 3$, the probability becomes $\approx 0.97$.  As we have discussed in Section \ref{sec:smallmarkets}, players' decision to stay or to match early imposes externalities on each other. As the number of players increases, their interdependent externalities become difficult to reconcile simultaneously, leading to chaos.

Without referring to the specific offer game played in the first period, this result implies that a small uncertainty introduced into the matching market will result in dramatic consequences. In applications, this would render institutional matching markets in the second period highly unlikely to survive in the long term. In a setting with complementarities between types, Theorem \ref{thm:chaos} also implies that inefficient outcomes are highly likely to occur.

The full proof of Theorem 1 is rather involved, but it is of value to outline the roadmap of the proof and discuss the underlying idea below.  Complete proofs can be found in the Appendix.

\subsection{Roadmap to the Proof of Theorem 1}
\label{sec:roadmap}

We first make an observation that is crucial to all of our analysis.  Let $L$ be the an ordered list of agents who stay to the second period, $m$ and $w$ be the types of a pair of equally-ranked agents, such that $(m,w) \in L$. Let $w_+$ and $w_-$ be the women just above and below $w$ in $L$. Man $m$'s expected payoff from waiting for the second period, $U(m, L)$, is given in equation \ref{eq:payoff_U}. It is easy to use integration by parts to verify that $m$ has strict incentive to match early with the woman of type $w$, or $U(m,L)<w$, if and only if:
\begin{equation}
\label{eq:intbyparts_man}
(1-F(m)) \int_{w_-}^w G(x) \, dx > F(m) \int_{w}^{w_+} (1-G(x)) \, dx.
\end{equation}

The LHS of \ref{eq:intbyparts_man} can be interpreted as $m$'s weighted downside in the second period, while the RHS is his weighted upside.  Thus, $m$ has incentive to match early with $w$ if and only if his weighted downside is bigger than his weighted upside in period 2. 

Analogously, let $m_+$ and $m_-$ be the men just above and below $m$ in $L$. Then woman $w$'s payoff in the second period is given by equation \ref{eq:payoff_V} (cf.\ Equation \ref{eq:payoff_V}), and she has strict incentive to match early with $m$ if and only if:

\begin{equation}
\label{eq:intbyparts_woman}
(1-G(w)) \int_{m_-}^m F(x) \, dx > G(w) \int_{m}^{m_+} (1-F(x)) \, dx.
\end{equation}


{\bf Step 1. Reducing to assortative early matching}: The first step of the proof is to convert any non-assortative, pairwise-stable early matching to an assortative, pairwise-stable early matching.  This is accomplished by two lemmas.

\begin{lemma}
Suppose that early matching $\mu$ is pairwise stable for a list of types.  Then, men and women who match early according $\mu$ must come in blocks; that is, let $I_m = \{i \in \{1, \ldots, n\} : \mu(i) \neq \emptyset\}$ be the set of ranks of men who match early, and $I_w = \{i \in \{1, \ldots, n\} : i \not \in \{\mu(1), \ldots, \mu(n)\} \}$ the set of ranks of women who match early, then we have $I_m = I_w$.
\end{lemma}
\begin{proof}
Suppose that $I_m \neq I_w$.  Let $j$ be the maximum element of the non-empty set $I_m \setminus I_w \cup I_w \setminus I_m$.  Without loss of generality suppose that $j \in I_m \setminus I_w$, i.e., man of rank $j$ is matched early, but woman of rank $j$ is not.  Since $j$ is the maximum element, man $j$ must be matched early to a woman ranked lower than $j$ ($\mu(j) < j$), because all women above $j$ are either ``taken'' (matched early with some man above rank $j$) or not matching early.  However, man of rank $j$ cannot have incentive to match early with woman of rank $\mu(j)$, because by not matching early with her she will still be available in the second period, while woman of rank $j$ is also available, and man of rank $j$ can be ``bumped'' down at most one place by the new agents, so woman of rank $\mu(j)$ is his lower bound in the second period.
\end{proof}

The above lemma guarantees that couples who do match early must be ``trapped'' between men and women of the same ranking who stay for the second period.  Therefore, we are justified in using Equation~\ref{eq:intbyparts_man} and \ref{eq:intbyparts_woman} to examine agents' incentive to match early.

\begin{lemma}
\label{lemma:uncross}
Suppose that for types $m_4 > m_3 > m_2 > m_1$ and $w_4 > w_3 > w_2 > w_1$, men $m_4$ and $m_1$ and woman $w_4$ and $w_1$ wait for the second period, while $m_3$ and $w_2$ have incentives to match early with each other, and likewise for $m_2$ and $w_3$.  Then, $m_3$ and $w_3$ also have incentives to match early with each other, and likewise for $m_2$ and $w_2$.
\end{lemma}
\begin{proof}
By (\ref{eq:intbyparts_man}), we have
\begin{equation*}
(1-F(m_3)) \int_{w_1}^{w_2} G(x) \, dx > F(m_3) \int_{w_2}^{w_4} (1-G(x)) \, dx,
\end{equation*}
which implies that (by changing $w_2$ to $w_3$) 
\begin{equation*}
(1-F(m_3)) \int_{w_1}^{w_3} G(x) \, dx > F(m_3) \int_{w_3}^{w_4} (1-G(x)) \, dx,
\end{equation*}
i.e., $m_3$ has incentive to match early with $w_3$, and that (by changing $m_3$ to $m_2$)
\begin{equation*}
(1-F(m_2)) \int_{w_1}^{w_2} G(x) \, dx > F(m_2) \int_{w_2}^{w_4} (1-G(x)) \, dx,
\end{equation*}
i.e., $m_2$ has incentive to match early with $w_2$.

Likewise, $w_3$ has incentive to match early with $m_3$, and $w_2$ has incentive to match early with $m_2$.
\end{proof}

Here is an intuitive story of why a pairwise-stable cross match ($m_2 - w_3$ and $m_3 - w_2$) can be converted to a pairwise-stable assortative match ($m_2 - w_2$, $m_3 - w_3$) in the above lemma: if woman $w_3$ is happy with $m_2$ (which is true, since $w_3 - m_2$ is pairwise stable), then of course she should be happy with a better man $m_3$.  And likewise man $m_3$ must be happy with woman $w_3$.  The trickier part is the lower couple: why should man $m_2$ be happy with woman $w_2$?  The answer is that since man $m_3$, who is in strictly better position than man $m_2$, is happy with woman $w_2$ (since $m_3 - w_2$ is pairwise stable), then the lower man $m_2$ must be happy with woman $w_2$ as well.  Similarly for woman $w_2$ with man $m_2$.

Converting a non-assortative, pairwise-stable early matching (call it $\mu_1$) to an assortative, pairwise-stable early matching is now transparent: among men and women who match early according to $\mu_1$, find the lowest cross match, and ``uncross'' the match as it is done in Lemma~\ref{lemma:uncross}, and leave all other early matches in $\mu_1$ unchanged.  Call the resulting early matching $\mu_2$.  By Lemma~\ref{lemma:uncross}, the early matching $\mu_2$ must be pairwise stable.  Clearly, $\mu_2$ has exactly one less cross match than $\mu_1$.  Now repeated the same procedure on $\mu_2$.  Since each time we get one less cross match while preserving pairwise stability, the procedure eventually terminates with an assortative early matching that is pairwise stable.

For concreteness, we illustrate this procedure in Figure~\ref{fig:convert}.

\begin{figure}[ht]
\begin{center}
	\includegraphics[scale=0.8]{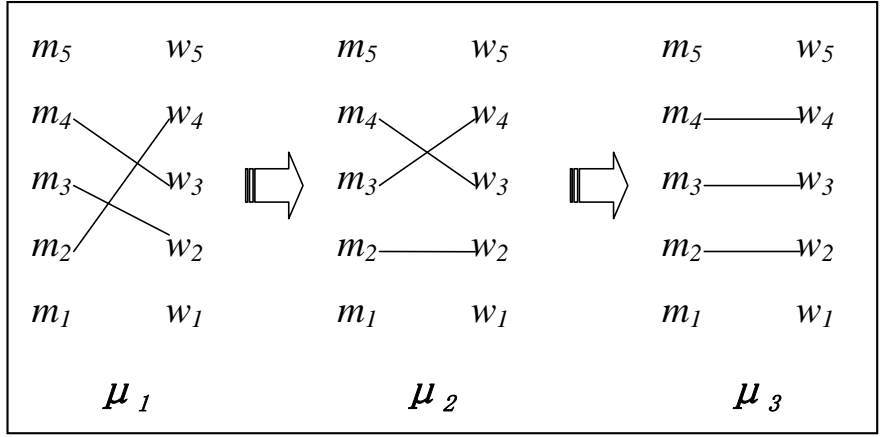}
	\caption{Converting non-assortative, pairwise-stable early matching to assortative, pairwise-stable early matching.  A line between $m_i$ and $w_j$ indicates that $m_i$ and $w_j$ match early.  $\mu_{i-1} \Rightarrow \mu_i$ means that if $\mu_{i-1}$ is pairwise-stable, then $\mu_i$ is pairwise-stable as well.} \label{fig:convert}
\end{center}
\end{figure}

Therefore, if a list of types admits a pairwise-stable early matching, it \emph{must} admit an assortative, pairwise-stable early matching.  Incidentally, we have the following curious property of assortative, pairwise-stable early matching.

\begin{proposition}
\label{prop:uniqueness}
Any list of types admits at most one assortative, pairwise-stable early matching.
\end{proposition}

The proof of the proposition is included in the appendix.  The proposition explains that our seemingly crude approach of using union bound to estimate the probability of non-chaotic realizations in the Appendix is in fact the right thing to do; however, the proposition will not play any direct role in our proof.

The advantage of working with assortative, pairwise-stable early matching is that agents' incentive consideration is fully captured by Equation~\ref{eq:intbyparts_man} and \ref{eq:intbyparts_woman}, which are greatly more tractable than working directly with $U(m, L)$ and $V(w, L)$.  It is informative and helpful for Step 2 to write out the definition of pairwise-stable, assortative early matching in terms of Equation (\ref{eq:intbyparts_man}) and (\ref{eq:intbyparts_woman}).

Let us represent an assortative early matching by the ranks of couples who stay to the second period: given an ordered list of ranks $\{i_l\}_{1 \leq l \leq L} \subseteq \{1, \ldots, n\}$, where $i_L > i_{L-1} > \ldots > i_1$, we have $\mu(i_l) = \emptyset$.  Every couple $(m_i, w_i)$ not on the list matches early with each other, i.e., $i \not \in \{i_l: 1 \leq l \leq L\} \Rightarrow \mu(i) = i$. 
 \begin{definition}
 \label{def:assortativestable}
 For an ordered list of first-period types $C=\{ (m_i, w_i) \}_{1 \leq i \leq n}$, an ordered list of ranks $\{i_l\}_{1 \leq l \leq L}$ is a \emph{pairwise-stable assortative arrangement (of waiting/matching early)} if
 \begin{enumerate}
 \item for couple $(m_{i_l}, w_{i_l})$ who waits ($1 \leq l \leq L$), we have either 
$$(1-F(m_{i_l})) \int_{w_{i_{l-1}}}^{w_{i_l}} G(x) \, dx \leq F(m_{i_l}) \int_{w_{i_l}}^{w_{i_{l+1}}} (1-G(x)) \, dx,$$
or
$$(1-G(w_{i_l})) \int_{m_{i_{l-1}}}^{m_{i_l}} F(x) \, dx \leq G(w_{i_l}) \int_{m_{i_l}}^{m_{i_{l+1}}} (1-F(x)) \, dx,$$
 \item for couple $(m_i, w_i)$ who matches early ($i \not \in \{i_1, \ldots, i_L\}$, so $i_l < i < i_{l+1}$ for some $l$), we have 
$$(1-F(m_{i})) \int_{w_{i_{l}}}^{w_{i}} G(x) \, dx > F(m_i) \int_{w_{i}}^{w_{i_{l+1}}} (1-G(x)) \, dx,$$
and
$$(1-G(w_{i})) \int_{m_{i_{l}}}^{m_{i}} F(x) \, dx > G(w_i) \int_{m_{i}}^{m_{i_{l+1}}} (1-F(x)) \, dx,$$
 \end{enumerate}
 where we use the convention that $i_0 = 0$, $i_{L+1} = n+1$, $m_0 = \underline{m}$, $w_0 = \underline{w}$, $m_{n+1} = \overline{m}$, $w_{n+1} = \overline{w}$.
\end{definition}

We note that woman of rank $i_l$ will never have incentive to match early with man of rank $i_{l'}$, for $l > l'$.  Therefore, in point (1) of the above definition we only check that man and woman of the same ranking who stay to the second period do not have joint incentive to match early.


In Step 2 to 4 we will only work with assortative arrangement.

{\bf Step 2. Local chaos}: The second step is to set up a local version of Theorem~\ref{thm:chaos}: suppose \emph{exogenously} (this assumption will be removed in Step 4) that couples of ranks $\lfloor n p \rfloor$ and $\lfloor n p \rfloor + r$ (where $p \in (0, 1)$ and integer $r \geq 2$ are fixed ex-ante) stay to the second period, what is the probability that couples in between them will experience chaos, i.e., cannot find a pairwise-stable assortative arrangement of waiting/matching early?

Fix a equally-ranked couple $(m, w)$ in between ranks $\lfloor n p \rfloor$ and $\lfloor n p \rfloor + r$.  Let $w_+$ and $w_-$ be the women just above and below $w$ who stay to the second period in an assortative arrangement.  Notice that $w_+ \leq w_{\lfloor n p \rfloor + r}$ and $w_- \geq w_{\lfloor n p \rfloor}$, because by assumption couples of ranks $\lfloor n p \rfloor$ and $\lfloor n p \rfloor + r$  stay to the second period.

As $n$ gets large, $w_{\lfloor n p \rfloor + r} - w_{\lfloor n p \rfloor}$ becomes (in probability) close to zero, and thus so does $w_+ - w_-$.  Therefore, Equation~\ref{eq:intbyparts_man}, the incentive for man $m$ to match early with $w$, becomes approximately
\begin{equation*}
(1-F(m)) G(w) (w-w_-) \, dx > F(m) (1-G(w)) (w_+ - w).
\end{equation*}

Now, since man $m$ and woman $w$ are ranked between $\lfloor n p \rfloor$ and $\lfloor n p \rfloor + r$, both $F(m)$ and $G(w)$ must converge (in probability) to $p \in (0, 1)$ as $n$ gets large.  Therefore, dividing out $p(1-p)$ on both sides, the above equation becomes
\begin{equation*}
w-w_- > w_+ - w.
\end{equation*}

Analogously, Equation~\ref{eq:intbyparts_woman}, the incentive for woman $w$ to match early with $m$, becomes approximately
\begin{equation*}
m-m_- > m_+ - m,
\end{equation*}
where $m_+$ and $m_-$ are the men just above and below $m$ who stay to the second period.

The above consideration motivates the following definition (compare with Definition~\ref{def:assortativestable}).  Let $\alpha_i = m_{\lfloor n p \rfloor + i} - m_{\lfloor n p \rfloor + i - 1}$ and $a_i = w_{\lfloor n p \rfloor + i} - w_{\lfloor n p \rfloor + i - 1}$ be the gaps between consecutive types:

\begin{definition}
\label{def:localstable}
Given an unordered list of gaps $\{ (\alpha_i, a_i) \}_{ 1 \leq i \leq r}$, an ordered list of indices $\{i_l\}_{1 \leq l \leq L} \subseteq \{ 1, \ldots, r-1 \}$ is a \emph{pairwise-stable assortative arrangement} if
\begin{enumerate}
\item for a couple of rank $\lfloor n p \rfloor + i_l$ who waits ($1 \leq l \leq L$), we have either $\sum_{j = 1+i_{l-1}}^{i_{l}} \alpha_j \leq \sum_{j = 1+i_l}^{i_{l+1}} \alpha_j$ or $\sum_{j = 1+i_{l-1}}^{i_{l}} a_j \leq \sum_{j = 1+i_l}^{i_{l+1}} a_j$,
\item for a couple of rank $\lfloor n p \rfloor + i$ who matches early ($i \not \in \{i_1, \ldots, i_L\}$, so $i_l < i < i_{l+1}$ for some $l$), we have $\sum_{j = 1+i_l}^{i} \alpha_j > \sum_{j = 1+i}^{i_{l+1}} \alpha_j$ and $\sum_{j = 1+i_l}^{i} a_j > \sum_{j = 1+i}^{i_{l+1}} a_j$,
\end{enumerate}
where $i_0 = 0$ and $i_{L+1} = r$.
\end{definition}

Finally, under some mild regularity conditions on the distributions $F$ and $G$, we may treat the gaps $\alpha_i$ and $a_i$'s as i.i.d.\ exponential random variables.  

Let $\pi(r)$ be the probability that gaps $\{ (\alpha_i, a_i) \}_{ 1 \leq i \leq r}$ admit a pairwise-stable assortative arrangement, given that gaps $\alpha_i$ and $a_i$'s are i.i.d.\ exponential random variables.  $\pi(r)$ is the limit, as $n$ tends to infinity, of the probability that couples in between ranks $\lfloor n p \rfloor$ and $\lfloor n p \rfloor + r$ admit a pairwise-stable assortative arrangement, given that couples of ranks $\lfloor n p \rfloor$ and $\lfloor n p \rfloor + r$ exogenously stay to the second period.

{\bf Step 3. Recursive argument}:  The third step is to show that $\pi(r)$ converges to 0 as $r$ tends to infinity.  This can be accomplished with a recursive argument exploiting the i.i.d.\ gaps.

Let region A cover positions between $\lfloor n p \rfloor$ and $\lfloor n p \rfloor + r$; region B covers positions between $\lfloor n p \rfloor$ and $\lfloor n p + r/2 \rfloor$; and region C covers positions between $\lfloor n p + r/2\rfloor$ and $\lfloor n p \rfloor + r$.  By construction, region A is the union of region B and C.

If couples in region A settle on a pairwise-stable assortative arrangement, then the truncated arrangement for couples in region B must be pairwise stable, and likewise for the truncated arrangement for couples in region C.  Since the gaps are i.i.d.\ exponential, couples in region B admitting a pairwise-stable assortative arrangement is an event independent of couples in region C admitting a pairwise-stable assortative arrangement.  Therefore, a recursive inequality in the spirit of $\pi(r) \leq \pi(r/2)^2$ must hold, which immediately implies that $\pi(r)$ goes to $0$ as $r$ tends to infinity.

The previous argument is not exactly correct, because we implicitly made the unjustified assumption that the couple of rank $\lfloor n p + r/2\rfloor$ stay to the second period.  Nevertheless, with high probability that in every pairwise-stable assortative arrangement \emph{some} couple around $\lfloor n p + r/2\rfloor$ would stay to the second period, and we simply sum over all possibilities.  In the end we get a weaker recursive inequality:

\begin{equation}
\label{eq:pi}
\pi(r) \leq \frac{13}{9} \pi(r/2)^2 + \epsilon,
\end{equation}
where $\epsilon>0$ is a small number.  

Inequality (\ref{eq:pi}), together with some exactly computed values of $\pi(r)$ when $r$ is small, implies that $\lim_{r \rightarrow \infty} \pi(r) = 0$.

{\bf Step 4. Wrapping up}: The final step is to remove the assumption that couples of ranks $\lfloor n p \rfloor$ and $\lfloor n p \rfloor + r$ exogenously stay to the second period.

For a fixed integer $s$ that is sufficiently large, we claim that in any pairwise-stable assortative arrangement it is highly unlikely that every couple in between ranks $\lfloor n p \rfloor - s$ and $\lfloor n p \rfloor$ chooses to match early, because lots of \emph{consecutive} couples matching early creates a large upside (the right-hand side of (\ref{eq:intbyparts_man}) and (\ref{eq:intbyparts_woman})), which tempts agents who match early to deviate by waiting for the second period and contradicts the stability of the arrangement.  

Therefore with high probability every pairwise-stable assortative arrangement has a couple (call them couple $x$) ranked between $\lfloor n p \rfloor$ and $\lfloor n p \rfloor - s$ staying to the second period, and likewise with high probability every pairwise-stable assortative arrangement has a couple (call them couple $y$) ranked between  $\lfloor n p \rfloor+r$ and $\lfloor n p \rfloor +r+s$.  Couples $x$ and $y$, by construction, are of at least $r$ ranks apart, so as $n$ gets big, the analysis in Step 2 and 3 implies that the probability that couples in between $x$ and $y$ admitting a pairwise-stable assortative arrangement is small, since we can choose any large value of $r$ before sending $n$ to infinity.  

Therefore, the probability that first-period types admitting a pairwise-stable assortative arrangement tends to 0 as $n$ gets big.  By Step 1, this means that the probability that first-period types admitting a pairwise-stable early matching tends to 0 as $n$ gets big.

\section{Unraveling} \label{sec:unraveling}

We now turn to studying agents incentive when all other agents stay in the market for the second period. This analysis of unraveling incentives is a useful benchmark, as having all agents be matched together assortatively is the efficient choice for agents when there are complementarities in types.

In this section we hold $k$ fixed, but no longer assume $k=1$. For the latter case, we have seen that this choice of strategy by agents will not be pairwise-stable with a probability which tends to 1 with $n$, and we conjecture that this also the case for any fixed $k$. Despite this, it is of interest to ask how far this choice of strategies is from an equilibrium --- does the probability that a pair of agents wants to deviate away from this proposed set of strategies tend to 0 with $n$? Does it depend on the type distributions $F$ and $G$? Are agents in different parts of the distribution more or less likely to deviate than others?  And how many agents would have incentive to deviate?

To answer these questions, we look at the incentives of agents to match early, assuming that all other agents indeed wait for the second period. Specifically, we focus on deviations by pairs of equally ranked men and women. These are obviously a strict subset of the possible deviations agents can choose. However, if we can get a lower bound on the probability that such deviations are profitable, this same bound will serve as a lower bound on the probability that \emph{any} deviation will be profitable.

Formally, suppose that $L$ is the a random rank-ordered list of $n$ pairs of types present on the market in the first period, with a distribution determined by the distributions $F$ and $G$. For the $j$-th ranked couple in $L$, let us define the \emph{probability of unraveling} $q_{j,n}$ to be the \emph{ex-ante} probability of both agents in the couple strictly wanting to match early, assuming everyone else stays on to the second period. We are interested in the asymptotic behavior of $q_{j_n,n}$, where $(j_n)_{n=1}^{\infty}$ is a sequence of positions.




We can now present the main result regarding unraveling in this perturbed market.

\begin{theorem} \label{thm:unraveling}
Fix any $k \geq 1$:
\begin{enumerate}

\item for any fixed percentile $p \in (0, 1)$, if the densities $f$ and $g$ are positive and continuous at $x$ and $y$, respectively, where $F(x) = p = G(y)$, then the ex-ante probability of unraveling for the couple at the $p$-th percentile obeys: $$\lim_{n \rightarrow \infty} q_{\lfloor n p \rfloor,n} = \frac{1}{4}$$

\item if the densities $f$ and $g$ are positive and continuous everywhere, then for any $\epsilon > 0$, with probability tending to 1 as $n \rightarrow \infty$, a fraction of at least $(\frac{1}{8k} - \epsilon)$ of all pairs have strict incentive to match early.
\end{enumerate}
\end{theorem}

Theorem \ref{thm:unraveling} gives a definitive answer to unraveling question. The first part of the theorem says that under mild regularity conditions on $F$ and $G$, a couple at any percentile $0<p<1$ will tend to prefer matching early, given that all the others are staying in the market, with an ex-ante probability of $1/4$. Thus, the incentives for unraveling for almost all pairs do not go to 0 as markets get thick, and this behavior is \emph{distribution free}. Note that this ex-ante unraveling probability is independent of agents' place within the distribution, as long as they located in the interior of the support.

The intuition for this part of the theorem draws on Step 2 of the roadmap for Theorem~\ref{thm:chaos}, in the case of $k=1$ --- in the limit, a man prefers to leave the market early if his downside from waiting, which is the distance between the type of the woman of his ranking and the woman directly beneath her, is bigger than his upside, which is the distance between the types of the woman of his ranking and the woman directly above her. For any regular distribution, the probability of this happening tends to $1/2$. A symmetric argument holds for women, and since these two calculations are independent, the unraveling probability for the pair converges to $1/4$.

The second part of Theorem \ref{thm:unraveling} extends the first and gives an asymptotic lower bound on the \emph{portion of the population} with simultaneous strict incentives to unravel. This lower bound is strictly positive, meaning that as the market gets thicker a large percentage of the population is likely to unravel given the benchmark of all agents waiting till the second period. The intuition here draws on the the size of the market ---  when the market grows thick, the realizations of ranked pairs far away from one another tend to be independent of each other. Thus, if each pair stands a chance of $1/4$ to want to leave the market early, then in probability, a positive percentage of pairs would like to unravel simultaneously. 

The next proposition gives a complete account of the unraveling probabilities of lowest and highest couples.  As in Theorem~\ref{thm:unraveling}, these asymptotic probabilities do not depend on the fine details of distributions $F$ and $G$.

\begin{proposition}\label{prop:unraveling}
Fix any $k \geq 1$:
\begin{enumerate}
\item Suppose that the densities $f$ and $g$ are continuous and positive at $\overline{m}$ and $\overline{w}$, respectively. Then the probability of unraveling for the $r$-th highest couple satisfies: 
$$\lim_{n\to\infty} q_{n+1-r,n}=\zeta_r \; > \; 0 \qquad,$$ 
with:
\begin{equation*}
\zeta_r = \Pr \left(\begin{array}{c} 2 (\alpha + \beta) c > (2 a + b) b, \\ 2 (a + b) \gamma > (2 \alpha + \beta) \beta \end{array}\right),
\end{equation*}
where $b,c,\beta,\gamma$ are independent exponential (with mean $1$) random variables, and $a,\alpha$ are independent Gamma $\Gamma(r-1,1)$ random variables\footnote{The distribution of a Gamma random variable $\Gamma(i,c)$ is defined by the pdf $h(x)=x^{i-1}\frac{e^{-\frac{x}{c}}}{c^r (i-1)!}$, $i \geq 1$.  We define a $\Gamma(0,1)$ random variable to be the constant 0.  $\Gamma(i,c)$ is the sum of $i$ i.i.d.\ exponential (with mean $c$) random variables.}.

\item Suppose that the densities $f$ and $g$ are continuous and positive at $\underline{m}$ and $\underline{w}$, respectively. Then the probability of unraveling for the $r$-th lowest couple satisfies:
$$\lim_{n\to\infty}q_{r,n}=\eta_r \; > \;0 \qquad,$$ 
with:
\begin{equation*}
\eta_r = \Pr \left(\begin{array}{c} (2 a + b) b > 2(\alpha + \beta) c, \\ 2 (\alpha + \beta) \beta > 2(a + b) \gamma \end{array}\right),
\end{equation*}
where $b,c,\beta,\gamma$ are independent exponential (with mean $1$) random variables, and $a,\alpha$ are independent Gamma $\Gamma(r-1,1)$ random variables.

\item $\displaystyle\lim_{r \rightarrow \infty} \zeta_r = 1/4 = \displaystyle\lim_{r \rightarrow \infty} \eta_r$.
\end{enumerate}
\end{proposition}

Proposition \ref{prop:unraveling} reveals that the lowest and highest couples in the population also have a strictly positive asymptotic probability of unraveling. To get a sense of these probabilities, charted below in Table \ref{table} are the values of $\zeta_r,\eta_r$ for $r\leq 10$. These values are strictly lower than $1/4$, the unraveling probability for couples in the middle of the distribution, but converge to $1/4$ as $r$ tends to infinity.

\begin{table}[ht]
\begin{center}
\begin{tabular}{|c | c | c |}
\hline
$r$ & $\zeta_r$ & $\eta_r$ \\
\hline
1 & 0.217602 & 0.0760076 \\
\hline
2 & 0.218555 & 0.127098 \\
\hline
3 & 0.223682 & 0.154687 \\
\hline
4 & 0.227859& 0.170272 \\
\hline
5 & 0.23102 & 0.177269 \\
\hline
6 & 0.233436 & 0.181989 \\
\hline
7 & 0.23533 & 0.186633 \\
\hline
8 & 0.236848 & 0.189689 \\
\hline
9 & 0.238072 & 0.198074 \\
\hline
10 & 0.239086 & 0.189451 \\
\hline
\end{tabular}
\end{center}
\caption{Values of $\zeta_r$ and $\eta_r$.} \label{table}
\end{table}

Theorem~\ref{thm:unraveling} and Proposition~\ref{prop:unraveling} make an unambiguous empirical prediction: both high ranked agents and low ranked agents are less likely to unravel than median ranked agents, while high ranked agents are more likely to unravel than low ranked agents.  Interestingly, this is broadly consistent with the pattern of unraveling in law-clerk matching market [TODO - need citation].

We note that Theorem~\ref{thm:unraveling} and Proposition \ref{prop:unraveling} are not merely results in ``asymptopia'': Figure~\ref{fig:finitesample} plots the unraveling probabilities, $\{ q_{i:25} \}_{1 \leq i \leq 25}$, for $n=25$ and $k=1$ and uniform distributions $F=G=\U[0,1]$; the probabilities are computed by Monte Carlo simulations.  The probabilities in the plot comes quite close to the asymptotic probabilities predicted by Proposition \ref{prop:unraveling}.

	\begin{figure}[ht]
	\begin{center}
		\includegraphics{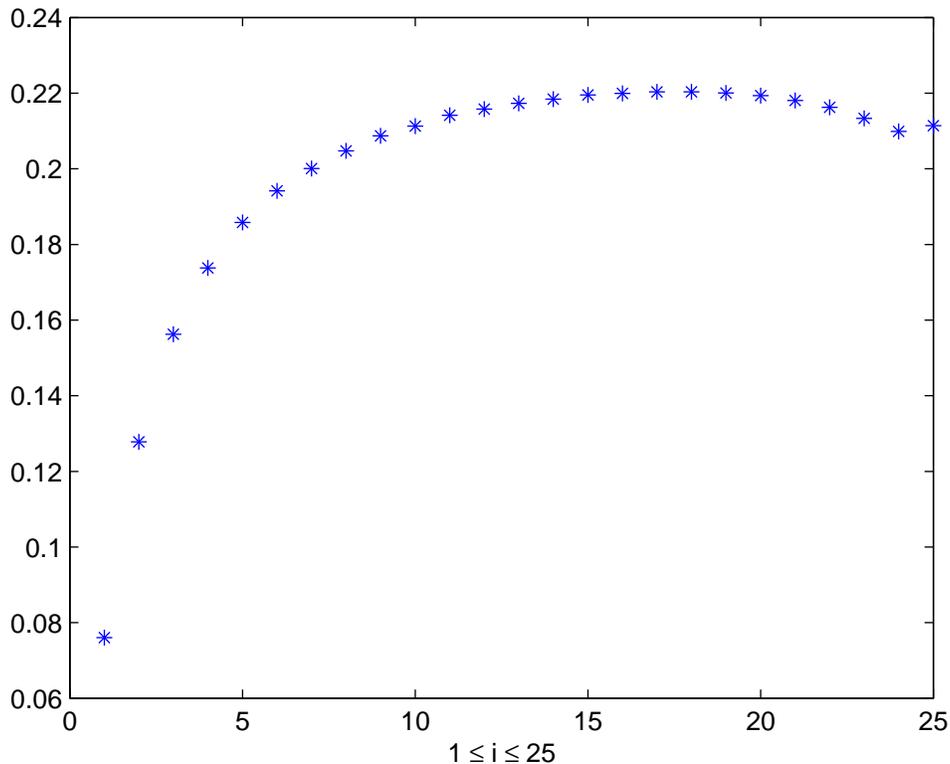}
		\caption{Plot of unraveling probabilities, $q_{i:25}$, for $n=25$ and $k=1$ and uniform distribution $F=G=\U[0,1]$.} 
\label{fig:finitesample}
	\end{center}
	\end{figure}

\section{Conclusion}\label{sec:conclusion}

This paper shows that matching markets may be very sensitive to small uncertainties and institutional details. Beyond unraveling, where a proportion of the agents in the market choose to match early \emph{in equilibrium}, we identify chaos as a dominant force in these perturbed markets. Under chaos, where essentially the game where agents choose when to match has no pure-strategy equilibrium, centralized and decentralized matching markets are likely to collapse. This paper thus serves as message to market designers, emphasizing the point that small details in matching mechanisms are of great importance to the future health of the market.

We conjecture that the negative results obtained in this paper are not unique to the specific modeling choice made here -- having the information perturbation be in the form of new entrants into the market. Alternative modeling choices for noise, in the form of agents leaving the market unexpectedly before the central match is held, or having agents' types be perturbed by shocks, should yield the similar instability results. These different types of noise may be more fitting to some applications, but moreover, proving the validity of our results to these settings would strengthen the point that the sensitivities of matching markets to institutional details are indeed an inherent feature of these markets.

\newpage

\section{Appendix}

\subsection{Proof of Theorem 1}

In Step 1 of the Roadmap (Section~\ref{sec:roadmap}), we argue that for the purpose of analyzing the prevalence of chaotic realization (first-period types that do not admit any pairwise-stable early matching), it is without loss of generality to restrict to pairwise-stable, assortative early matching (i.e., pairwise-stable assortative arrangement, cf.\ Definition~\ref{def:assortativestable}).  Therefore, in this section we will only work with assortative arrangement.

For any index set $I \subseteq \{1, \ldots, n\}$, let $C_I \subseteq \{ (m_i, w_i)_{1 \leq i \leq n} \mid m_n \geq m_{n-1} \geq \ldots \geq m_1, w_n \geq w_{n-1} \geq \ldots \geq w_1 \}$ be the set of first-period types that admit $I$ as a pairwise-stable assortative arrangement.

A restatement of Theorem~\ref{thm:chaos} is that under the stated regularity condition, we have  
\begin{equation*}
\lim_{n \rightarrow \infty} \Pr_n \left( \bigcup_{I \subseteq \{1, \ldots, n\}} C_I \right) = 0,
\end{equation*}
where $\Pr_n$ is the measure associated with the order statistics of $n$ i.i.d.\ random variables of distribution $F$ (men's types) and of $n$ i.i.d.\ random variables of distribution $G$ (women's types).

By assumption (i), fix a $p \in (0, 1)$ such that the densities $f$ and $g$ are continuous and positive at $x$ and $y$, respectively, where $F(x) = G(y) = p$.

For any two integers $s > 0$ and $t > 0$, let 
\begin{equation*}
\calI = \{ I \mid I \subseteq \{1, \ldots, n\} \}.
\end{equation*}
\begin{equation}
\label{eq:Cjl}
\calI_{j, l} = \left \{ I \in \calI \mid \begin{array}{l}
\min( \{ i \in I : i \geq \lfloor n p \rfloor + t \} ) = \lfloor n p \rfloor + t + j, \text{ and} \\
\max( \{ i \in I : i \leq \lfloor n p \rfloor - t \} ) = \lfloor n p \rfloor - t - l
\end{array} \right \}.
\end{equation}
\begin{equation*}
\calI' = \left \{ I \in \calI \mid \begin{array}{l}
\min( \{ i \in I : i \geq \lfloor n p \rfloor + t \} ) > \lfloor n p \rfloor + t + s, \text{ or} \\
\max( \{ i \in I : i \leq \lfloor n p \rfloor - t \} ) < \lfloor n p \rfloor - t -s
\end{array} \right \}.
\end{equation*}

Clearly, 
\begin{equation*}
\calI = \calI' \cup \bigcup_{\substack{0 \leq j \leq s \\ 0 \leq l \leq s}} \calI_{j, l}.
\end{equation*}

We can further divide $\calI'$ into $\calI'_1$ and $\calI'_2$: $\calI = \calI'_1 \cup \calI'_2$, where
\begin{equation*}
\calI'_1 = \left \{ I \in \calI \mid 
\min( \{ i \in I : i \geq \lfloor n p \rfloor + t \} ) > \lfloor n p \rfloor + t + s
\right \},
\end{equation*}
\begin{equation*}
\calI'_2 = \left \{ I \in \calI \mid 
\max( \{ i \in I : i \leq \lfloor n p \rfloor - t \} ) < \lfloor n p \rfloor - t -s
\right \}.
\end{equation*}

Let $C(\calI'') = \bigcup_{I \in \calI''} C_I$ for $\calI'' \subseteq \calI$.  Then,
\begin{equation*}
\Pr_n \left( \bigcup_{I \in \calI} C_I \right) \leq \Pr_n(C(\calI'_1)) + \Pr_n(C(\calI'_2)) + \sum_{\substack{0 \leq j \leq s \\ 0 \leq l \leq s}} C(\calI_{j, l})
\end{equation*}

Therefore, 
\begin{align}
\label{eq:bound}
\limsup_{n \rightarrow \infty} \Pr_n \left( \bigcup_{I \in \calI} C_I \right) \leq & \limsup_{n \rightarrow \infty} \Pr_n(C(\calI'_1)) + \limsup_{n \rightarrow \infty} \Pr_n(C(\calI'_2)) \notag \\
& + \sum_{\substack{0 \leq j \leq s \\ 0 \leq l \leq s}} \limsup_{n \rightarrow \infty} \Pr_n(C(\calI_{j, l})).
\end{align}

In the next two subsections we show that for a fixed $s$, $\limsup_{n \rightarrow \infty} \Pr_n(C(\calI_{j, l}))$ goes to 0 as $t$ goes to infinity for any $j$ and $l$, at a rate independent of $s$ (Section~\ref{sec:localchaos}); and that for a fixed $t$, $\limsup_{n \rightarrow \infty} \Pr_n(C(\calI'_1))$ and $\limsup_{n \rightarrow \infty} \Pr_n(C(\calI'_2))$ go to 0 as $s$ goes to infinity, at a rate independent of $t$ (Section~\ref{sec:gap}).  This implies that by choosing $s$ and $t$ sufficiently large, we can make the left hand side of (\ref{eq:bound}), which is independent of $s$ and $t$, as close to zero as we want.  Thus, the left hand side of (\ref{eq:bound}) must be exactly zero, which proves Theorem~\ref{thm:chaos}.

\subsubsection{Local Chaos}
\label{sec:localchaos}

We first bound $\limsup_{n \rightarrow \infty} \Pr_n(C(\calI_{j, l}))$.

\begin{definition}
Given $p \in (0, 1)$, integer $r \geq 2$ and an ordered list of types $\{ (m_i, w_i) \}_{ \lfloor n p \rfloor \leq i \leq \lfloor n p \rfloor + r}$, an ordered list of indices $\{i_l\}_{1 \leq l \leq L} \subseteq \{ \lfloor n p \rfloor+1, \ldots, \lfloor n p \rfloor+r-1 \}$ is a \emph{pairwise-stable assortative arrangement} if
\begin{enumerate}
\item for a couple of rank $i_l$ who waits ($1 \leq l \leq L$), we have either 
$$(1-F(m_{i_l})) \int_{w_{i_{l-1}}}^{w_{i_l}} G(x) \, dx \leq F(m_{i_l}) \int_{w_{i_l}}^{w_{i_{l+1}}} (1-G(x)) \, dx,$$
or
$$(1-G(w_{i_l})) \int_{m_{i_{l-1}}}^{m_{i_l}} F(x) \, dx \leq G(w_{i_l}) \int_{m_{i_l}}^{m_{i_{l+1}}} (1-F(x)) \, dx,$$
\item for a couple of rank $i$ who matches early ($i \not \in \{i_1, \ldots, i_L\}$, so $i_l < i < i_{l+1}$ for some $l$), we have 
$$(1-F(m_{i})) \int_{w_{i_{l}}}^{w_{i}} G(x) \, dx > F(m_i) \int_{w_{i}}^{w_{i_{l+1}}} (1-G(x)) \, dx,$$
and
$$(1-G(w_{i})) \int_{m_{i_{l}}}^{m_{i}} F(x) \, dx > G(w_i) \int_{m_{i}}^{m_{i_{l+1}}} (1-F(x)) \, dx,$$
\end{enumerate}
where $i_0 = \lfloor n p \rfloor$ and $i_{L+1} = \lfloor n p \rfloor + r$.
\end{definition}

For any index set $I \subseteq \{ \lfloor n p \rfloor+1, \ldots, \lfloor n p \rfloor+r-1 \}$, let $C_I \subseteq \{ (m_i, w_i)_{\lfloor n p \rfloor \leq i \leq \lfloor n p \rfloor + r} \mid m_{\lfloor n p \rfloor + r} \geq  \ldots \geq m_{\lfloor n p \rfloor}, w_{\lfloor n p \rfloor + r} \geq \ldots \geq w_{\lfloor n p \rfloor} \}$ be the set of couple types that admit $I$ as a pairwise-stable assortative arrangement.


For any index set $I \subseteq \{ 1, \ldots, r-1 \}$, let $G_I \subseteq \R_+^{2 r}$ be the set of gaps that admit $I$ as a pairwise-stable assortative arrangement (see Definition~\ref{def:localstable} in Section~\ref{sec:roadmap}), where $\R_+$ is the set of non-negative real numbers.

\begin{proposition}
\label{prop:convergetopi}
Fix $p \in (0, 1)$ and integer $r \geq 2$.  If the densities $f$ and $g$ are positive and continuous at $x$ and $y$, respectively, where $F(x) = p = G(y)$.  Then
\begin{equation*}
\lim_{n \rightarrow \infty} \Pr_n \left( \bigcup_{I \subseteq \{ \lfloor n p \rfloor+1, \ldots, \lfloor n p \rfloor+r-1 \}} C_I \right) = \hat{\Pr}_r \left( \bigcup_{I \subseteq \{ 1, \ldots, r-1 \}} G_I \right),
\end{equation*}
where $\hat{\Pr}_r$ is the measure of $2r$ i.i.d.\ exponential (with mean 1) random variables.
\end{proposition}
\begin{proof}
The proof relies on the well-known Slutsky's Theorem and is similar to the proof of Theorem~\ref{thm:unraveling} and 
Proposition~\ref{prop:unraveling}.
\end{proof}

\begin{proposition}
\label{prop:GIdisjoint}
For $I \neq I' \subseteq \{ 1, \ldots, r-1 \}$, we have $G_I \cap G_{I'} = \emptyset$.
\end{proposition}
\begin{proof}
This is a special case of Proposition~\ref{prop:uniqueness}.
\end{proof}

\begin{proposition}
\label{prop:piGI}
Suppose that $I = \{i_l \mid 1 \leq l \leq L\}$, where $r > i_L > \ldots > i_1 > 0$.  Then
\begin{align*}
 \hat{\Pr}_r(G_I) = & \hat{\Pr}_r \left( \forall i \not \in I \text{ such that } i_l < i < i_{l+1},
\begin{array}{l}
\sum_{j=1+i}^{i_{l+1}} \alpha_j < \sum_{j=1+i_l}^i \alpha_j \text{ and}\\
\sum_{j=1+i}^{i_{l+1}} a_j < \sum_{j=1+i_l}^i a_j
\end{array}
\right) \\
& \cdot \hat{\Pr}_r \left( \forall l \in \{1, \ldots, L\},
\begin{array}{l}
\sum_{j=1+i_l}^{i_{l+1}} \alpha_j \geq \sum_{j=1+i_{l-1}}^{i_l} \alpha_j \text{ or}\\
\sum_{j=1+i_l}^{i_{l+1}} a_j \geq \sum_{j=1+i_{l-1}}^{i_l} a_j
\end{array}
\right)
\end{align*}
\begin{align*}
= \frac{1}{4^{r-L-1}} \int_{\substack{x_1, y_1, \ldots, \\ x_{L+1}, y_{L+1}} \geq 0} &
\ind \left( \begin{array}{c}
(x_2 \geq x_1 \text{ or } y_2 \geq y_1)  \text{ and} \\
\ldots \text{ and} \\ 
(x_{L+1} \geq x_{L} \text{ or } y_{L+1} \geq y_{L})
\end{array}
\right)
\prod_{l=1}^{L+1} \frac{e^{-x_l -y_l} (x_l y_l)^{i_l - i_{1-1} - 1}} {((i_l - i_{l-1} - 1)!)^2} \\
& d(x_1, y_1, \ldots, x_{L-1}, y_{L-1})
\end{align*}
where $i_0 = 0$ and $i_{L+1} = r$.
\end{proposition}
\begin{proof}
This exploits the memorylessness of exponential distribution.  Proof to be added later.
\end{proof}

\begin{proposition}
\label{prop:limitpir}
$\pi(r) := \hat{\Pr}_r \left( \bigcup_{I \subseteq \{ 1, \ldots, r-1 \}} G_I \right)$ tends to 0 as $r \rightarrow \infty$.
\end{proposition}

\begin{proof}
We have for $l \leq \lfloor n/2 \rfloor$:
\begin{align}
\label{eq:pirecursive}
\pi(r) \leq &  \sum_{j=1}^{\lfloor n/2 \rfloor} \left( \frac{1}{4} \right)^{l+1+j-1} + \pi ( \lfloor  n/2 \rfloor) \pi( n - \lfloor  n/2 \rfloor ) \\
& + \sum_{i=1}^l \left( 
\sum_{j=1}^{l} \left( \frac{1}{4} \right)^{i+j-1} \pi \left( \lfloor n/2 \rfloor - j \right) \pi( n - \lfloor n/2 \rfloor - i ) + \sum_{j=l+1}^{ \lfloor n/2 \rfloor} \left( \frac{1}{4} \right)^{i+j-1} 
\right) \notag \\
\leq & \left( \frac{1}{4} \right)^l \frac{7}{9} + \frac{13}{9} \max_{-1 \leq j \leq l} \pi( \lfloor n/2 \rfloor - j)^2. \notag
\end{align}

Using Proposition~\ref{prop:piGI} and Proposition~\ref{prop:GIdisjoint}, we can directly calculate $\pi(r)$ when $r$ is small.  Our calculations imply that $\pi(7) \approx  0.595$ and that $\pi(r) < 0.595$ for $8 \leq r \leq 17$.  When $r = 18$, let $l = 2$, then (\ref{eq:pirecursive}) implies that $\pi(18) \leq 13/9 \cdot 0.595^2 + 7/144 = 0.559981 < 0.595$.  In general, when $y$ is between the two roots $0.0526089$ and $0.639699$ of equation $13/9 x^2 + 7/144 - x =0$, then $13/9 y^2 + 7/144 < y$.

For $18 < r$, we can use progressively larger values of $l$ in (\ref{eq:pirecursive}) to conclude that $\pi(r)$ tends to $0$ as $r \rightarrow \infty$.
\end{proof}

Now going back to bounding (\ref{eq:bound}): by Proposition \ref{prop:convergetopi} we have
\begin{equation*}
\limsup_{n \rightarrow \infty} \Pr_n(C(\calI_{j, l})) \leq \pi(2t+j+l)
\end{equation*}
because positions $\lfloor n p \rfloor + t + j$ and $\lfloor n p \rfloor - t - l$ are of distance $2t+j+l$ apart (cf.\ Equation~\ref{eq:Cjl}).

Therefore, Proposition~\ref{prop:limitpir} proves that $\limsup_{n \rightarrow \infty} \Pr_n(C(\calI_{j, l}))$  goes to 0 as $t$ goes to infinity, for any fixed $s>0$, $0 \leq j \leq s$ and $0 \leq l \leq s$.

\subsubsection{Consecutively ranked couples matching early}
\label{sec:gap}

In this subsection we give the required bounds for $\limsup_{n \rightarrow \infty} \Pr_n(C(\calI'_1))$ and $\limsup_{n \rightarrow \infty} \Pr_n(C(\calI'_2))$.  For this step we need assumption (2) of the theorem.

By combining assumptions (i) and (ii), it's easy to verify that there exist an $\epsilon_0 > 0$ and finite and positive $\underline{a}$ and $\overline{a}$ such that $\underline{a} \leq f(m), g(w) \leq \overline{a}$ hold for all $(m, w)$ such that $f(m), g(w) \leq p + \epsilon_0$.  Let $\bar{p} = p + \epsilon_0$.

Let $\bar{t} = t+s+1$.  And let $w_0 = \underline{w}$ and $m_0 = \underline{m}$.

Summing over all possible positions ($i-1$ in the summation below) just below $\lfloor n p \rfloor + t$ in which a couple stays in a pairwise-stable assortative arrangement, we have:
\begin{align}
& \Pr_n(C(\calI'_1)) \notag \\
\leq & \sum_{i=1}^{\lfloor n p \rfloor + t} \Pr_n \left(
\begin{array}{c}
(1-F(m_i)) \int_{w_{i-1}}^{w_i} G(x) dx > F(m_i) \int_{w_i}^{w_{\lfloor n p \rfloor + \bar{t}}} (1-G(x)) dx, \\
(1-G(w_i)) \int_{m_{i-1}}^{m_i} F(x) dx > G(w_i) \int_{m_i}^{m_{\lfloor n p \rfloor + \bar{t}}} (1-F(x)) dx
\end{array}
\right) \notag \\
\leq & \sum_{i=\lfloor n p_0 \rfloor + 1}^{\lfloor n p \rfloor + t} \Pr_n \left(
\begin{array}{c}
\int_{w_{i-1}}^{w_i} G(x) dx > F(m_{\lfloor n p_0 \rfloor}) \int_{w_i}^{w_{\lfloor n p \rfloor + \bar{t}}} (1-G(x)) dx, \\
\int_{m_{i-1}}^{m_i} F(x) dx > G(w_{\lfloor n p_0 \rfloor}) \int_{m_i}^{m_{\lfloor n p \rfloor + \bar{t}}} (1-F(x)) dx
\end{array}
\right) \notag \\
     & + \sum_{i=1}^{\lfloor n p_0 \rfloor} \Pr_n \left(
\begin{array}{c}
\int_{w_{i-1}}^{w_i} G(x) dx > F(m_i) \int_{w_{\lfloor n p_0 \rfloor}}^{w_{\lfloor n p \rfloor + \bar{t}}} (1-G(x)) dx, \\
\int_{m_{i-1}}^{m_i} F(x) dx > G(w_i) \int_{m_{\lfloor n p_0 \rfloor}}^{m_{\lfloor n p \rfloor + \bar{t}}} (1-F(x)) dx
\end{array}
\right) \notag \\
\label{eq:preterm1}
\leq  & \sum_{i=\lfloor n p_0 \rfloor + 1}^{\lfloor n p \rfloor + t} \Pr_n \left(
\begin{array}{c}
(w_i - w_{i-1}) > F(m_{\lfloor n p_0 \rfloor}) (1-G(w_{\lfloor n p \rfloor + \bar{t}})) (w_{\lfloor n p \rfloor + \bar{t}} - w_i), \\
(m_i - m_{i-1}) > G(w_{\lfloor n p_0 \rfloor}) (1-F(m_{\lfloor n p \rfloor + \bar{t}})) (m_{\lfloor n p \rfloor + \bar{t}} - m_i)
\end{array}
\right) \\
\label{eq:preterm2}
     & + \sum_{i=1}^{\lfloor n p_0 \rfloor} \Pr_n \left(
\begin{array}{c}
G(w_i) (w_i - w_{i-1}) > F(m_i) (1-G(w_{\lfloor n p \rfloor + \bar{t}})) (w_{\lfloor n p \rfloor + \bar{t}} - w_{\lfloor n p_0 \rfloor}) \\
F(m_i) (m_i - m_{i-1}) > G(w_i) ((1-F(m_{\lfloor n p \rfloor + \bar{t}}))) (m_{\lfloor n p \rfloor + \bar{t}} - m_{\lfloor n p_0 \rfloor})
\end{array}
\right), \notag \\
\end{align}

where $p_0 < p$ is arbitrary.

Our goal is then to bound (\ref{eq:preterm1}) and (\ref{eq:preterm2}).  Before continuing let's work out a large deviation inequality in the context of order statistics that will immensely simplify our analysis.

\begin{lemma}[Chernoff] 
\label{lemma:chernoff}
$\Pr_n(|F(m_i) - i/n| \geq \epsilon) \leq 2 \exp(-2\epsilon^2 n)$ for any $\epsilon > 0$.
\end{lemma}
\begin{proof}
Without loss of generality suppose that $F(m_i) = u_i$, where $u_i$ is the $i$th order statistics of $n$ i.i.d. uniform (over $[0, 1]$) random variables.

Clearly,
\begin{equation*}
\Pr_n(|u_i - i/n| \geq \epsilon) \leq \Pr_n(u_i \geq i/n + \epsilon) + \Pr_n(u_i \leq i/n - \epsilon).
\end{equation*}

By definition, we have
\begin{equation*}
\Pr_n(u_i \leq i/n - \epsilon) = \Pr_n \left( \sum_{j = 1}^n \ind(z_j \leq i/n - \epsilon) \geq i \right)
\end{equation*}
where $z_1, \ldots, z_n$ are $n$ i.i.d.\ uniform $[0, 1]$ random variables.

We now apply a standard Chernoff bound to i.i.d.\ random variables $\ind(z_j \leq i/n - \epsilon)$'s (e.g., \cite{AlonSpencer}, Theorem A.1.4):

\begin{align*}
\Pr_n \left( \sum_{j = 1}^n \ind(z_j \leq i/n - \epsilon) \geq i \right) & = \Pr_n \left( \sum_{j = 1}^n (\ind(z_j \leq i/n - \epsilon) - (i/n - \epsilon)) \geq n \epsilon \right) \\
& \leq \exp(-2 \epsilon^2 n).
\end{align*}

Similarly, 
\begin{align*}
\Pr_n(u_i \geq i/n + \epsilon) & = \Pr_n(u_i > i/n + \epsilon) \\
& = \Pr_n \left( \sum_{j = 1}^n \ind(z_j \leq i/n + \epsilon) < i \right) \\
 & = \Pr_n \left( \sum_{j = 1}^n (\ind(z_j \leq i/n + \epsilon) - (i/n + \epsilon)) < -n \epsilon \right) \\
& \leq \exp(-2 \epsilon^2 n).
\end{align*}
\end{proof}

Returning to bounding (\ref{eq:preterm1}):

\begin{align*}
& \Pr_n \left(
\begin{array}{c}
(w_i - w_{i-1})  > F(m_{\lfloor n p_0 \rfloor}) (1-G(w_{\lfloor n p \rfloor + \bar{t}})) (w_{\lfloor n p \rfloor + \bar{t} } - w_i), \\
(m_i - m_{i-1})  > G(w_{\lfloor n p_0 \rfloor}) (1-F(m_{\lfloor n p \rfloor + \bar{t}})) (m_{\lfloor n p \rfloor + \bar{t} } - m_i)
\end{array}
\right) \\
\leq & \Pr_n \left(
\begin{array}{c}
(w_i - w_{i-1})  > F(m_{\lfloor n p_0 \rfloor}) (1-G(w_{\lfloor n p \rfloor + \bar{t}})) (w_{\lfloor n p \rfloor + \bar{t}} - w_i), 
G(w_{\lfloor n p \rfloor + \bar{t}}) \leq \bar{p}, F(m_{\lfloor n p_0 \rfloor}) \geq p_0 - \epsilon \\
(m_i - m_{i-1})  > G(w_{\lfloor n p_0 \rfloor}) (1-F(m_{\lfloor n p_0 \rfloor + \bar{t}})) (m_{\lfloor n p \rfloor + \bar{t}} - m_i),
F(m_{\lfloor n p \rfloor + \bar{t}}) \leq \bar{p}, G(w_{\lfloor n p_0 \rfloor}) \geq p_0 - \epsilon
\end{array}
\right) \\
& + \Pr_n(G(w_{\lfloor n p \rfloor + \bar{t}}) > \bar{p}) + \Pr_n(F(m_{\lfloor n p \rfloor + \bar{t}}) > \bar{p}) +\Pr_n(F(m_{\lfloor n p_0 \rfloor}) < p_0 - \epsilon) + \Pr_n(G(w_{\lfloor n p_0 \rfloor}) < p_0 - \epsilon)
\notag \\
\leq & \Pr_n \left(
\begin{array}{c}
(w_i - w_{i-1})  > (p_0 - \epsilon) (1-\bar{p}) (w_{\lfloor n p \rfloor + \bar{t}} - w_i), 
G(w_{\lfloor n p \rfloor + \bar{t}}) \leq \bar{p}, F(m_{\lfloor n p_0 \rfloor}) \geq p_0 - \epsilon \\
(m_i - m_{i-1})  > (p_0 - \epsilon) (1-\bar{p}) (m_{\lfloor n p \rfloor + \bar{t}} - m_i),
F(m_{\lfloor n p \rfloor + \bar{t}}) \leq \bar{p}, G(w_{\lfloor n p_0 \rfloor}) \geq p_0 - \epsilon
\end{array}
\right) \\
& + \Pr_n(G(w_{\lfloor n p \rfloor + \bar{t}}) > \bar{p}) + \Pr_n(F(m_{\lfloor n p \rfloor + \bar{t}}) > \bar{p}) +\Pr_n(F(m_{\lfloor n p_0 \rfloor}) < p_0 - \epsilon) + \Pr_n(G(w_{\lfloor n p_0 \rfloor}) < p_0 - \epsilon)
\notag \\
\leq & \Pr_n \left(
\begin{array}{c}
(\overline{a}/\underline{a})  (v_i - v_{i-1})  > (p_0 - \epsilon) (1-\bar{p}) (v_{\lfloor n p \rfloor + \bar{t}} - v_i), \\
(\overline{a}/\underline{a}) (u_i - u_{i-1})  > (p_0 - \epsilon) (1-\bar{p}) (u_{\lfloor n p \rfloor + \bar{t}} - u_i)
\end{array}
\right) \\
& + \Pr_n(G(w_{\lfloor n p \rfloor + \bar{t}}) > \bar{p}) + \Pr_n(F(m_{\lfloor n p \rfloor + \bar{t}}) > \bar{p}) +\Pr_n(F(m_{\lfloor n p_0 \rfloor}) < p_0 - \epsilon) + \Pr_n(G(w_{\lfloor n p_0 \rfloor}) < p_0 - \epsilon),
\end{align*}

where $u_1, \ldots, u_n$ are the order statistics of $n$ i.i.d.\ uniform (over $[0, 1]$) random variables (let $u_0 = 0$); and likewise $v_1, \ldots, v_n$ are the order statistics of $n$ i.i.d.\ uniform (over $[0, 1]$) random variables, independent of $u$'s (let $v_0 = 0$). 

Since $(m_i, w_i)_{1 \leq i \leq n}$ has the exact same distribution as $(F^{-1}(u_i), G^{-1}(v_i))_{1 \leq i \leq n}$, it is w.l.o.g.\ to assume that $m_i = F^{-1}(u_i)$, $w_i = G^{-1}(v_i)$.

By assumption, we have $0 < \underline{a} \leq f(m), g(w) \leq \overline{a} < \infty$ for all $m \in [\underline{m}, F^{-1}(\bar{p})]$ and $w \in [\underline{w}, G^{-1}(\bar{p})]$.  Thus, by the Mean Value Theorem, $(u_i - u_j)/\overline{a} \leq m_i - m_j \leq  (u_i - u_j)/\underline{a}$ and $(v_i - v_j)/\overline{a} \leq w_i - w_j \leq  (v_i - v_j)/\underline{a}$ hold, for any $i > j$.

By Lemma~\ref{lemma:chernoff}, we have

\begin{align*}
\lim_{n \rightarrow \infty} n ( & \Pr_n(G(w_{\lfloor n p \rfloor + \bar{t}}) > \bar{p}) + \Pr_n(F(m_{\lfloor n p \rfloor + \bar{t}}) > \bar{p}) \\
& + \Pr_n(F(m_{\lfloor n p_0 \rfloor}) < p_0 - \epsilon) + \Pr_n(G(w_{\lfloor n p_0 \rfloor}) < p_0 - \epsilon) ) = 0,
\end{align*}

therefore, to bound $\limsup_{n \rightarrow \infty} \Pr_n(C(\calI'_1))$ we can replace (\ref{eq:preterm1}) by
\begin{equation}
\label{eq:term1}
 \sum_{i=\lfloor n p_0 \rfloor + 1}^{\lfloor n p \rfloor + t} \Pr_n \left(
\begin{array}{c}
(\overline{a}/\underline{a}) (v_i - v_{i-1})  > (p_0-\epsilon) (1-\bar{p}) (v_{\lfloor n p \rfloor + \bar{t}} - v_i), \\
(\overline{a}/\underline{a}) (u_i - u_{i-1})  > (p_0-\epsilon) (1-\bar{p}) (u_{\lfloor n p \rfloor + \bar{t}} - u_i)
\end{array}
\right).
\end{equation}

Similarly, we can replace (\ref{eq:preterm2}) by
\begin{equation}
\label{eq:term2}
\sum_{i=1}^{\lfloor n p_0 \rfloor} \Pr_n \left(
\begin{array}{c}
(\overline{a}/\underline{a}) v_i (v_i - v_{i-1}) > u_i (1-\bar{p}) (\bar{p} - p_0), \\
(\overline{a}/\underline{a}) u_i (u_i - u_{i-1}) > v_i (1-\bar{p}) (\bar{p} - p_0)
\end{array}
\right)
\end{equation}

The following lemma takes care of (\ref{eq:term2}):
\begin{lemma}
\begin{equation*}
\lim_{n \rightarrow \infty} \sum_{i=1}^{\lfloor n p_0 \rfloor} \Pr_n \left(
\begin{array}{c}
(\overline{a}/\underline{a}) v_i (v_i - v_{i-1}) > u_i (1-\bar{p}) (\bar{p} - p_0), \\
(\overline{a}/\underline{a}) u_i (u_i - u_{i-1}) > v_i (1-\bar{p}) (\bar{p} - p_0)
\end{array}
\right) = 0.
\end{equation*}
\end{lemma}
\begin{proof}
For any $1 \leq i \leq \lfloor n p_0 \rfloor$, we have:
\begin{align*}
& \Pr_n \left(
\begin{array}{c}
(\overline{a}/\underline{a}) v_i (v_i - v_{i-1}) > u_i (1-\bar{p}) (\bar{p} - p_0), \\
(\overline{a}/\underline{a}) u_i (u_i - u_{i-1}) > v_i (1-\bar{p}) (\bar{p} - p_0)
\end{array}
\right) \\
\leq & \Pr_n \left(
(v_i - v_{i-1}) (u_i - u_{i-1}) > \left( (\underline{a}/\overline{a}) (1-\overline{p})(\overline{p}-p_0) \right)^2
\right) \\
\leq & \Pr_n \left(
(v_i - v_{i-1}) > \left( (\underline{a}/\overline{a}) (1-\overline{p})(\overline{p}-p_0) \right)^2
\right) \\
\leq & \left( 1 - \left( (\underline{a}/\overline{a}) (1-\overline{p})(\overline{p}-p_0) \right)^2 \right)^n.
\end{align*}
The last line uses an elementary property of the order statistics of uniform distribution.
\end{proof}

The following lemma is a classic result in statistics:
\begin{lemma}
$(1-u_n, u_n - u_{n-1}, u_{n-1} - u_{n-2}, \ldots, u_2 - u_1, u_1)$ has the same distribution as $(x_1 / x, x_2 / x, \ldots, x_{n+1} / x)$, where $x_1, \ldots, x_{n+1}$ are i.i.d.\ exponential (with mean 1) random variables and $x = \sum_{i = 1}^{n+1} x_i$.
\end{lemma}

\begin{lemma}
Suppose that $x_1, \ldots x_{l}, y$ are i.i.d.\ exponential (with mean 1) random variables.  Then for any $a>0$, 
\begin{equation*}
\Pr \left( a \sum_{i = 1}^{l} x_i \geq y \right) = ( 1+a )^{-l}.
\end{equation*}
\end{lemma}

These two lemmas imply:
\begin{align*}
& \limsup_{n \rightarrow \infty} \sum_{i=\lfloor n p_0 \rfloor + 1}^{\lfloor n p \rfloor + t} \Pr_n \left(
\begin{array}{c}
(\overline{a}/\underline{a}) (v_i - v_{i-1})  > (p_0 - \epsilon) (1-\bar{p}) (v_{\lfloor n p \rfloor + \bar{t}} - v_i), \\
(\overline{a}/\underline{a}) (u_i - u_{i-1})  > (p_0 - \epsilon) (1-\bar{p}) (u_{\lfloor n p \rfloor + \bar{t}} - u_i)
\end{array}
\right) \\
\leq & \limsup_{n \rightarrow \infty}  \sum_{i=\lfloor n p_0 \rfloor + 1}^{\lfloor n p \rfloor + t} \Pr_n \left(
\begin{array}{c}
(\overline{a}/\underline{a}) (v_i - v_{i-1})  > (p_0 - \epsilon) (1-\bar{p}) (v_{\lfloor n p \rfloor + \bar{t}} - v_i), \\
(\overline{a}/\underline{a}) (u_i - u_{i-1})  > (p_0 - \epsilon) (1-\bar{p}) (u_{\lfloor n p \rfloor + \bar{t}} - u_i)
\end{array}
\right) \\
\leq &  \limsup_{n \rightarrow \infty}  \sum_{i=\lfloor n p_0 \rfloor + 1}^{\lfloor n p \rfloor + t} 
\left( 1 + \frac{\underline{a} (p_0 - \epsilon) (1-\bar{p})}{\overline{a}} \right)^{- (\lfloor n p \rfloor + \bar{t} - i)} \left( 1 + \frac{\underline{a} (p_0 - \epsilon) (1-\bar{p})}{\overline{a}} \right)^{- (\lfloor n p \rfloor + \bar{t} - i)} \\
\leq & \left(1 - \left( 1 + \frac{\underline{a} (p_0 - \epsilon) (1-\bar{p})}{\overline{a}} \right)^{-1} \left( 1 + \frac{\underline{a} (p_0 - \epsilon) (1-\bar{p})}{\overline{a}} \right)^{-1} \right)^{-1} \\
& \cdot \left( \left( 1 + \frac{\underline{a} (p_0 - \epsilon) (1-\bar{p})}{\overline{a}}  \right) \left( 1 + \frac{\underline{a} (p_0 - \epsilon) (1-\bar{p})}{\overline{a}} \right) \right)^{-s}
\end{align*}

Therefore, we have
\begin{align}
\label{eq:bound1}
\limsup_{n \rightarrow \infty} \Pr_n(C(\calI'_1)) \leq & \left(1 - \left( 1 + \frac{\underline{a} (p_0 - \epsilon) (1-\bar{p})}{\overline{a}} \right)^{-1} \left( 1 + \frac{\underline{a} (p_0 - \epsilon) (1-\bar{p})}{\overline{a}} \right)^{-1} \right)^{-1} \\
& \cdot \left( \left( 1 + \frac{\underline{a} (p_0 - \epsilon) (1-\bar{p})}{\overline{a}}  \right) \left( 1 + \frac{\underline{a} (p_0 - \epsilon) (1-\bar{p})}{\overline{a}} \right) \right)^{-s}. \notag
\end{align}

By exact same argument, we have
\begin{align}
\label{eq:bound2}
\limsup_{n \rightarrow \infty} \Pr_n(C(\calI'_2)) \leq & \left(1 - \left( 1 + \frac{\underline{a} (p_0 - \epsilon) (1-\bar{p})}{\overline{a}} \right)^{-1} \left( 1 + \frac{\underline{a} (p_0 - \epsilon) (1-\bar{p})}{\overline{a}} \right)^{-1} \right)^{-1} \\
& \cdot \left( \left( 1 + \frac{\underline{a} (p_0 - \epsilon) (1-\bar{p})}{\overline{a}}  \right) \left( 1 + \frac{\underline{a} (p_0 - \epsilon) (1-\bar{p})}{\overline{a}} \right) \right)^{-s}. \notag
\end{align}

\subsection{Proof of Theorem~\ref{thm:unraveling}}
\subsubsection{Part (1)}
\label{sec:unraveling_part1}
For a weakly-increasing list of $i$ numbers $x_1 \leq \ldots \leq  x_i$ and a distribution $D$, let $h_+(\{x_1, \ldots, x_i\}, j, D)$ be the expected $i$th lowest value among $x_1, \ldots, x_i$ and $j$ i.i.d.\ $D$ distributed random variables.

Likewise, for a weakly-increasing list of $i$ numbers $x_{-i} \leq \ldots \leq  x_{-1}$ and a distribution $D$, let $h_-(\{x_{-i}, \ldots, x_{-1}\}, j, D)$ be the expected $i$th highest value among $x_1, \ldots, x_i$ and $j$ i.i.d.\ $D$ distributed random variables.

Fix a $p \in (0, 1)$ such that $f$ and $g$ are continuous and positive at $\hat{m}$ and $\hat{w}$, respectively, where $F(\hat{m}) = G(\hat{w}) = p$.  We are interested in the asymptotic probability of unraveling of the $\lfloor n p \rfloor$th highest couple $(m_{\lfloor n p \rfloor}, w_{\lfloor n p \rfloor})$.

By staying to the second period, assuming everyone else does so as well, man $m_{\lfloor n p \rfloor}$'s expected payoff is

\begin{align*}
\sum_{0 \leq i < j \leq k} & {k \choose i} F(m_{\lfloor n p \rfloor})^i (1-F(m_{\lfloor n p \rfloor}))^{k-i} {k \choose j} G(w_{\lfloor n p \rfloor})^j (1-G(w_{\lfloor n p \rfloor}))^{k-j} \\
& \times h_-( \{ w_{\lfloor n p \rfloor - (j-i)}, \ldots, w_{\lfloor n p \rfloor - 1} \}, j, G(x | x \leq w_{\lfloor n p \rfloor}))
\end{align*}
\begin{align*}
+ \sum_{0 \leq j < i \leq k} & {k \choose i} F(m_{\lfloor n p \rfloor})^i (1-F(m_{\lfloor n p \rfloor}))^{k-i} {k \choose j} G(w_{\lfloor n p \rfloor})^j (1-G(w_{\lfloor n p \rfloor}))^{k-j} \\
& \times h_+( \{ w_{\lfloor n p \rfloor + 1}, \ldots, w_{\lfloor n p \rfloor + (i - j)} \}, k-j, G(x | x \geq w_{\lfloor n p \rfloor}))
\end{align*}
\begin{align*}
+ \sum_{0 \leq i \leq k} & {k \choose i} F(m_{\lfloor n p \rfloor})^i (1-F(m_{\lfloor n p \rfloor}))^{k-i} {k \choose i} G(w_{\lfloor n p \rfloor})^i (1-G(w_{\lfloor n p \rfloor}))^{k-i} \\
& \times w_{\lfloor n p \rfloor}.
\end{align*}

Comparing this to $w_{\lfloor n p \rfloor}$, we see that man $m_{\lfloor n p \rfloor}$ strictly prefers to match early with woman $w_{\lfloor n p \rfloor}$ if and only if:
\begin{align*}
n \sum_{0 \leq i < j \leq k} & {k \choose i} F(m_{\lfloor n p \rfloor})^i (1-F(m_{\lfloor n p \rfloor}))^{k-i} {k \choose j} G(w_{\lfloor n p \rfloor})^j (1-G(w_{\lfloor n p \rfloor}))^{k-j} \\
& \times (w_{\lfloor n p \rfloor} - h_-( \{ w_{\lfloor n p \rfloor - (j-i)}, \ldots, w_{\lfloor n p \rfloor - 1} \}, j, G(x | x \leq w_{\lfloor n p \rfloor})))
\end{align*}
\begin{align}
\label{eq:downsidevsupside}
> n \sum_{0 \leq j < i \leq k} & {k \choose i} F(m_{\lfloor n p \rfloor})^i (1-F(m_{\lfloor n p \rfloor}))^{k-i} {k \choose j} G(w_{\lfloor n p \rfloor})^j (1-G(w_{\lfloor n p \rfloor}))^{k-j} \notag \\
& \times ( h_+( \{ w_{\lfloor n p \rfloor + 1}, \ldots, w_{\lfloor n p \rfloor + (i - j)} \}, k-j, G(x | x \geq w_{\lfloor n p \rfloor})) - w_{\lfloor n p \rfloor} ),
\end{align}
where we have multiplied both sides by $n$.

Notice that
\begin{align*}
& G(w_{\lfloor n p \rfloor})^j (w_{\lfloor n p \rfloor} - h_-( \{ w_{\lfloor n p \rfloor - (j-i)}, \ldots, w_{\lfloor n p \rfloor - 1} \}, j, G(x | x \leq w_{\lfloor n p \rfloor}))) \\
= & G(w_{\lfloor n p \rfloor - (j-i)})^j (w_{\lfloor n p \rfloor} - w_{\lfloor n p \rfloor - (j-i)}) \\
& + (G(w_{\lfloor n p \rfloor})^j - G(w_{\lfloor n p \rfloor - (j-i) })^j) O_p(w_{\lfloor n p \rfloor} - w_{\lfloor n p \rfloor - (j-i)})
\end{align*}

Since 
\begin{equation*}
n (G(w_{\lfloor n p \rfloor})^j - G(w_{\lfloor n p \rfloor - (j-i) })^j) O_p(w_{\lfloor n p \rfloor} - w_{\lfloor n p \rfloor - (j-i)})
\end{equation*}
tends to 0 in probability, we can ignore this term.

Similarly, we can simplify
\begin{equation*}
(1-G(w_{\lfloor n p \rfloor}))^{k-j} ( h_+( \{ w_{\lfloor n p \rfloor + 1}, \ldots, w_{\lfloor n p \rfloor + (i - j)} \}, k-j, G(x | x \geq w_{\lfloor n p \rfloor})) - w_{\lfloor n p \rfloor} )
\end{equation*}
to
\begin{equation*}
(1-G(w_{\lfloor n p \rfloor + (i-j) }))^{k-j} (w_{\lfloor n p \rfloor + (i-j)} - w_{\lfloor n p \rfloor }).
\end{equation*}

We can change variables to order statistics of i.i.d.\ uniform $[0, 1]$ random variables: $m_i = F^{-1}(u_i), w_i = G^{-1}(v_i)$ (where by convention $u_{n+1} = v_{n+1} = 1$); since $g$ is continuous and positive at $\hat{w}$, the function $G^{-1}(v)$ is differentiable in a neighborhood $(p-\epsilon, p+\epsilon)$, $\epsilon > 0$, with derivative $\frac{1}{g(G^{-1}(v))}$.

Therefore, 
\begin{equation*}
\frac{w_{\lfloor n p \rfloor + (i-j)} - w_{\lfloor n p \rfloor }}{(v_{\lfloor n p \rfloor + (i-j)} - v_{\lfloor n p \rfloor })g(\hat{w})}
\end{equation*}
and
\begin{equation*}
\frac{w_{\lfloor n p \rfloor} - w_{\lfloor n p \rfloor - (j-i) }}{(v_{\lfloor n p \rfloor} - v_{\lfloor n p \rfloor - (j-i) })g(\hat{w})}
\end{equation*}
converge in probability to 1, while $G(w_{\lfloor n p \rfloor - i})$ and $F(m_{\lfloor n p \rfloor - i})$ converge in probability to $p$ for any fixed $i$.

Finally, it's well-known that (see \cite{pyke})
\begin{equation*}
(n(v_{\lfloor n p \rfloor + k} - v_{\lfloor n p \rfloor + k - 1 }), \ldots, n(v_{\lfloor n p \rfloor - k +1} - v_{\lfloor n p \rfloor - k }) ) \rightarrow_D (a_i)_{1 \leq i \leq 2k},
\end{equation*}
where $(a_i)_{-k \leq i \leq k}$ are $2k$ i.i.d.\ exponential (with mean 1) random variables, and the convergence is in distribution.

Therefore, by Slutsky's Theorem, the probability that (\ref{eq:downsidevsupside}) holds converges, as $n \rightarrow \infty$, to $1/2$.  Combining this analysis with a symmetric analysis for woman $w_{\lfloor n p \rfloor}$'s incentive to match early with man $m_{\lfloor n p \rfloor}$, it can be easily seen that the probability of unraveling for couple $(m_{\lfloor n p \rfloor}, w_{\lfloor n p \rfloor})$ converges to $1/4$.

\subsubsection{Part (2)}
\label{sec:unraveling_part2}
To be added later.

\subsection{Proof of Proposition~\ref{prop:unraveling}}

\subsubsection{Part (1) and (2)}
\label{sec:extremeunraveling}
We will prove only part (1); part (2) is completely analogous.

Suppose that $f$ and $g$ are positive and continuous at $\overline{m}$ and $\overline{w}$, respectively.  And fix a positive integer $r$.  We are interested in the asymptotic probability of unraveling of the $r$th highest couple $(m_{n-r+1}, w_{n-r+1})$.

As in Section~\ref{sec:unraveling_part1} (where the $h_+$ and $h_-$ functions are defined), man $m_{n-r+1}$ has strict incentive to match early with woman $w_{n-r+1}$ if and only if
\begin{align*}
n^2 \sum_{0 \leq i < j \leq k} & {k \choose i} F(m_{n-r+1})^i (1-F(m_{n-r+1}))^{k-i} {k \choose j} G(w_{n-r+1})^j (1-G(w_{n-r+1}))^{k-j} \\
& \times (w_{n-r+1} - h_-( \{ w_{n-r+1 - (j-i)}, \ldots, w_{{n-r+1} - 1} \}, j, G(x | x \leq w_{n-r+1})))
\end{align*}
\begin{align}
\label{eq:extincentiveman}
> n^2 \sum_{0 \leq j < i \leq k} & {k \choose i} F(m_{n-r+1})^i (1-F(m_{n-r+1}))^{k-i} {k \choose j} G(w_{n-r+1})^j (1-G(w_{n-r+1}))^{k-j}  \notag \\
& \times ( h_+( \{ w_{n-r+1 + 1}, \ldots, w_{n-r+1 + (i - j)} \}, k-j, G(x | x \geq w_{n-r+1})) - w_{n-r+1} ),
\end{align}
where we have multiplied both sides by $n^2$.

Likewise, woman $w_{n-r+1}$ has strict incentive to match early with man $m_{n-r+1}$ if and only if
\begin{align*}
n^2 \sum_{0 \leq i < j \leq k} & {k \choose i} G(w_{n-r+1})^i (1-G(w_{n-r+1}))^{k-i} {k \choose j} F(m_{n-r+1})^j (1-F(w_{n-r+1}))^{k-j} \\
& \times (m_{n-r+1} - h_-( \{ m_{n-r+1 - (j-i)}, \ldots, m_{{n-r+1} - 1} \}, j, F(x | x \leq m_{n-r+1})))
\end{align*}
\begin{align}
\label{eq:extincentivewoman}
> n^2 \sum_{0 \leq j < i \leq k} & {k \choose i} G(w_{n-r+1})^i (1-G(w_{n-r+1}))^{k-i} {k \choose j} F(m_{n-r+1})^j (1-F(m_{n-r+1}))^{k-j}  \notag \\
& \times ( h_+( \{ m_{n-r+1 + 1}, \ldots, m_{n-r+1 + (i - j)} \}, k-j, F(x | x \geq m_{n-r+1})) - m_{n-r+1} ).
\end{align}

We use the convention that $w_{n+i} = \overline{w}$ for all $1 \leq i \leq k$, so that $$h_+( \{ w_{n-r+1 + 1}, \ldots, w_{n-r+1 + k} \}, G(x | x \geq w_{n-r+1}))$$ is always sensibly defined.  Likewise, $m_{n+i} = \overline{m}$ for all $1 \leq i \leq k$.

Notice that in the LHS of (\ref{eq:extincentiveman}), all of the terms converge in probability to 0, except for $j = k$ and $i = k-1$.  Likewise, in the RHS of (\ref{eq:extincentiveman}), all of the terms converge in probability to 0, except for $j = k-1$ and $i = k$.  Therefore, we can concentrate on
\begin{align}
\label{eq:extincentive2}
& n^2  k F(m_{n-r+1})^{k-1} (1-F(m_{n-r+1})) G(w_{n-r+1})^k (w_{n-r+1} - h_-( \{ w_{n-r+1 - 1} \}, k, G(x | x \leq w_{n-r+1}))) \notag  \\
> & n^2 k F(m_{n-r+1})^k G(w_{n-r+1})^{k-1} (1-G(w_{n-r+1})) ( h_+( \{ w_{n-r+1 + 1} \}, 1, G(x | x \geq w_{n-r+1})) - w_{n-r+1} ).
\end{align}

We have
\begin{align*}
& (1-G(w_{n-r+1})) ( h_+( \{ w_{n-r+1 + 1} \}, 1, G(x | x \geq w_{n-r+1})) - w_{n-r+1} ) \\
= & (1-G(w_{n-r+1+1}))w_{n-r+1+1} + \int_{w_{n-r+1}}^{w_{n-r+2}} x g(x) \, dx  - (1-G(w_{n-r+1})) w_{n-r+1} \\
= & \int_{w_{n-r+1}}^{w_{n-r+2}} (1-G(x)) \, dx.
\end{align*}
and
\begin{align*}
& G(w_{n-r+1})^k (w_{n-r+1} - h_-( \{ w_{n-r+1 - 1} \}, k, G(x | x \leq w_{n-r+1}))) \\
= & G(w_{n-r})^{k} (w_{n-r+1}-w_{n-r}) + k G(w_{n-r})^{k-1} \int_{w_{n-r}}^{w_{n-r+1}} (w_{n-r+1}-x) g(x) \, dx \\
  & + O_p((G(w_{n-r+1})-G(w_{n-r}))^2) (w_{n-r+1}-w_{n-r}).\\
\end{align*}

It's easy to show that
\begin{align*}
& n^2 k F(m_{n-r+1})^{k-1} (1-F(m_{n-r+1})) \\
\cdot & \left( k G(w_{n-r})^{k-1} \int_{w_{n-r}}^{w_{n-r+1}} (w_{n-r+1}-x) g(x) \, dx  + O_p((G(w_{n-r+1})-G(w_{n-r}))^2)  (w_{n-r+1}-w_{n-r}) \right)
\end{align*}
converges in probability to $0$ as $n \rightarrow \infty$. 

Therefore, Equation~(\ref{eq:extincentive2}) becomes
\begin{align*}
& n^2 k F(m_{n-r+1})^{k-1} (1-F(m_{n-r+1})) G(w_{n-r})^k (w_{n-r+1}-w_{n-r}) \\
> & n^2 k F(m_{n-r+1})^k G(w_{n-r+1})^{k-1} \int_{w_{n-r+1}}^{w_{n-r+2}} (1-G(x)) \, dx.
\end{align*}

By changing variables to order statistics of i.i.d.\ uniform $[0, 1]$ random variables: $m_i = F^{-1}(u_i), w_i = G^{-1}(v_i)$ (where by convention $u_{n+1} = v_{n+1} = 1$), we have
\begin{align}
\label{eq:extincentive3}
& n^2  k (u_{n-r+1})^{k-1} (1-u_{n-r+1}) (v_{n-r})^k (G^{-1}(v_{n-r+1}) - G^{-1}(v_{n-r})) \notag \\
> & n^2 k (u_{n-r+1})^k (v_{n-r+1})^{k-1} \int_{v_{n-r+1}}^{v_{n-r+2}} (1-v) \frac{1}{g(G^{-1}(v))} \, dv,
\end{align}
since $g$ is continuous and positive at $\overline{w}$, the function $G^{-1}(v)$ is differentiable in a neighborhood $(1-\epsilon, 1)$, $\epsilon > 0$, with derivative $\frac{1}{g(G^{-1}(v))}$.

We have $u_{n-r+1}$, $v_{n-r+1}$ and $v_{n-r}$ converge to 1 in probability, and so do
\begin{equation*}
\frac{ G^{-1}(v_{n-r+1}) - G^{-1}(v_{n-r}) }{(v_{n-r+1} - v_{n-r})/g(\overline{w})}
\end{equation*}
and
\begin{equation*}
\frac{\int_{v = v_{n-r+1}}^{v_{n-r+2}} (1-v) \frac{1}{g(G^{-1}(v))} \, dv}{(2-v_{n-r+1}-v_{n-r+2})(v_{n-r+2} - v_{n-r+1})/(2g(\overline{w}))}.
\end{equation*}

On other other hand, 
\begin{equation*}
(n(1-u_{n}), \ldots, n(u_{n-r+1} - u_{n-r}), n(1-v_{n}), \ldots, n(v_{n-r+1} - v_{n-r}) ) \rightarrow_D
(\alpha_i, a_i)_{1 \leq i \leq r+1},
\end{equation*}
where $(\alpha_i, a_i)_{1 \leq i \leq r+1}$ are $2(r+1)$ i.i.d.\ exponential (with mean 1) random variables, and the convergence is in distribution.

Decomposing (\ref{eq:extincentive3}) into the above parts, and do the analogous analysis on (\ref{eq:extincentivewoman}), we conclude by Slutsky's Theorem that the probability that both (\ref{eq:extincentiveman}) and (\ref{eq:extincentivewoman}) hold converges, as $n$ tends to infinity, to
 \begin{equation*}
 \zeta_i = \Pr \left(\begin{array}{c} 2 (\alpha + \beta) c > (2 a + b) b, \\ 2 (a + b) \gamma > (2 \alpha + \beta) \beta \end{array}\right),
 \end{equation*}
where $\alpha = \sum_{i = 1}^{r-1} \alpha_i$, $\beta = \alpha_{r}$, $\gamma = \alpha_{r+1}$, $a = \sum_{i = 1}^{r-1} a_i$, $b = a_{r}$, and $c = a_{r+1}$.

\subsubsection{Part (3)}
By the weak law of large number, $\alpha/(r-1)$ and $a/(r-1)$ converge to 1 as $r$ tends to infinity.  On the other hand, $\beta c / (r-1)$, $b^2 / (r-1)$, $b \gamma / (r-1)$ and $\beta^2 / (r-1)$ converge in probability to 0.  The conclusion then follows by Slusky's Theorem.

\bibliography{unraveling-bibliography}

\end{document}